%% file: example_paper.tex
\documentclass{article}
\usepackage[accepted]{icml2018}
\usepackage{microtype}
\usepackage{graphicx}
\usepackage{subfigure}
\usepackage{booktabs} % for professional tables
\usepackage{hyperref}
\usepackage{breakurl}
\usepackage{amssymb}  %
\usepackage{amsmath}  % Math fonts
\usepackage{amsthm}
\usepackage{mathtools}
\usepackage{epstopdf}

\newtheorem{theorem}{Theorem}
\newtheorem{lemma}[theorem]{Lemma}

\newtheorem{fact}[theorem]{Fact}
\newtheorem{definition}[theorem]{Definition}

\newcommand{\R}{\mathbb{R}}
\DeclareMathOperator*{\E}{{\bf {E}}}

\DeclareMathOperator*{\argmin}{arg\,min}

\icmltitlerunning{Subspace Embedding and Linear Regression with Orlicz Norm}

\begin{document}

\twocolumn[
\icmltitle{Subspace Embedding and Linear Regression with Orlicz Norm}




\begin{icmlauthorlist}
\icmlauthor{Alexandr Andoni}{columbia}
\icmlauthor{Chengyu Lin}{columbia}
\icmlauthor{Ying Sheng}{columbia}
\icmlauthor{Peilin Zhong}{columbia}
\icmlauthor{Ruiqi Zhong}{columbia}
\end{icmlauthorlist}

\icmlaffiliation{columbia}{Computer Science Department, Columbia University, New York City, NY 10027, U.S.A.}
\icmlcorrespondingauthor{Peilin Zhong}{pz2225@columbia.edu}

\icmlkeywords{Machine Learning}

\vskip 0.3in
]



%\printAffiliationsAndNotice{}
% leave blank if no need to mention equal contribution
\printAffiliationsAndNotice{ \icmlEqualContribution} % otherwise use
%the standard text.

% Compact itemize and enumerate.  Note that they use the same counters and
% symbols as the usual itemize and enumerate environments.
\def\compactify{\itemsep=0pt \topsep=0pt \partopsep=0pt \parsep=0pt}
\let\latexusecounter=\usecounter
\newenvironment{enumerate*}
  {\def\usecounter{\compactify\latexusecounter}
   \begin{enumerate}}
  {\end{enumerate}\let\usecounter=\latexusecounter}
\newenvironment{itemize*}%
  {\begin{itemize}%
    \setlength{\itemsep}{0pt}%
    \setlength{\parskip}{0pt}}%
  {\end{itemize}}

\newcommand{\aanote}[1]{{\color{red}$\ll$\textsf{#1 --Alex}$\gg$\marginpar{\tiny\bf AA}}}
\newcommand{\Peilin}[1]{\textbf{\color{red}[Peilin: #1]\marginpar{\tiny\bf PZ}}}

\begin{abstract}
\input{abstract}
\end{abstract}

\input{intro}
\input{preli}

\input{subspace}

\input{applications}

\input{experiments}
\input{future}

\bibliography{refs}
\bibliographystyle{icml2018}

\input{arxiv_appendix}

\end{document}

%% file: abstract.tex
We consider a generalization of the classic linear regression problem
to the case when the loss is an Orlicz norm. An Orlicz norm is
parameterized by a non-negative convex function
$G:\mathbb{R}_+\rightarrow\mathbb{R}_+$ with $G(0)=0$: the Orlicz norm
of a vector $x\in\mathbb{R}^n$ is defined as
$
\|x\|_G=\inf\left\{\alpha>0\large\mid\sum_{i=1}^n G(|x_i|/\alpha)\leq 1\right\}.
$
We consider the cases where the function $G(\cdot)$ grows subquadratically. Our main result is based on a new oblivious embedding which embeds the
column space of a given matrix $A\in\mathbb{R}^{n\times d}$ with
Orlicz norm into a lower dimensional space with $\ell_2$
norm. Specifically, we show how to efficiently find an embedding
matrix $S\in\mathbb{R}^{m\times n},m<n$ such that $\forall
x\in\mathbb{R}^{d},\Omega(1/(d\log
n)) \cdot \|Ax\|_G\leq \|SAx\|_2\leq O(d^2\log n) \cdot \|Ax\|_G.$ By
applying this subspace embedding technique, we show an approximation
algorithm for the regression problem
$\min_{x\in\mathbb{R}^d} \|Ax-b\|_G$, up to a $O(d\log^2 n)$
factor. As a further application of our techniques, we show how to
also use them to improve on the algorithm for the $\ell_p$ low rank
matrix approximation problem for $1\leq p<2$.

%% file: intro.tex
\section{Introduction}
Numerical linear algebra problems play a significant role in machine
learning, data mining, and statistics. One of the most important such
problems is the regression problem, see
some recent advancements in, e.g., \cite{zjd16,bjk15,jt15,lwz14,dlfu13}. In
a linear regression problem, given a data matrix
$A\in\mathbb{R}^{n\times d}$ with $n$ data points $A^1,A^2,\cdots,A^n$
in $\mathbb{R}^{d}$ and the response vector $b\in\mathbb{R}^n$, the
goal is to find a set of coefficients $x^*\in\mathbb{R}^d$ such that
\begin{align}\label{eq:regression_prob}
x^*={\argmin}_{x\in\mathbb{R}^d}\ l(Ax-b),
\end{align}
where $l:\mathbb{R}^{n}\rightarrow \mathbb{R}_+$ is the loss
function. When $l(y) = \|y\|_2^2= \sum_{i=1}^n y_i^2,$ then the
problem is the classic least square regression problem
($\ell_2$ regression). While there has been extensive research on
efficient algorithms for solving $\ell_2$ regression, it
is not always a suitable loss function to use.

In many settings, alternative choices for a loss function lead to
qualitatively better solutions $x^*$. For example, a popular such
alternative is the least absolute deviation ($\ell_1$) regression ---
with $l(y) = \|y\|_1= \sum_{i=1}^n |y_i|$ --- which leads to
solutions that are more robust than those of $\ell_2$ regression
(see~\cite{l1l2regression, g05}.  In a nutshell, the $\ell_2$ regression
is suitable when the data contains Gaussian noise, whereas
$\ell_1$ --- when the noise is Laplacian or sparse.

A further popular class of loss functions $l(\cdot)$ arises
from {\em M-estimators}, defined as $l(y) = \sum_{i=1}^n M(y_i)$ where
$M(\cdot)$ is an M-estimator function (see~\cite{z97} for a list of
M-estimators). The benefit of (some) M-estimators is that they enjoy
advantages of both $\ell_1$ and $\ell_2$ regression. For example, when
$M(\cdot)$ is the Huber function~\cite{h64}, then the regression looks
like $\ell_2$ regression when $y_i$ is small, and looks like $\ell_1$
regression otherwise.  However, these loss functions come with a
downside: they depend on the scale, and rescaling the data may give a
completely different solution!

%% he downside of this loss
%% function is that $M^{-1}(\sum_{i=1}^n M(y_i))$ does not define a norm, because it is
%% not scale-invariant, i.e. $a M^{-1}(\sum_{i=1}^n M(y_i))\not= M^{-1}(\sum_{i=1}^n M(ay_i))$. This is not desirable,
\vspace{-0.15in}
\paragraph{Our contributions.}
We introduce a generic algorithmic technique for solving regression
for an entire class of loss functions that includes the aforementioned
examples, and in particular, a ``scale-invariant'' version of
M-estimators.
%Here, motivated by the advantages and deficiencies of M-estimators,
%, which has properties of both $\ell_1$ norm and $\ell_2$ norm.
Specifically, our class consists of loss functions $l(y)$ that are
Orlicz norms, defined as follows: given a non-negative convex function
$G:\mathbb{R}_+\rightarrow\mathbb{R}_+$ with $G(0)=0$, for
$x\in\mathbb{R}^n,$ we can define $\|x\|_G$ to be an Orlicz norm with
respect to $G(\cdot)$: $ \|x\|_G \triangleq
\inf\left\{\alpha>0\large\mid\sum_{i=1}^n G(|x_i|/\alpha)\leq
1\right\}.  $ Note that $\ell_p$ norm, for $p\in[1,\infty)$, is a
  special case of Orlicz norm with $G(x) = x^p$.  Another important
  example is the following ``scale-free'' version of M-estimator. Taking
  $f(\cdot)$ to be a Huber function, i.e.
\begin{align*}
f(x) = \left\{\begin{array}{ll}x^2/2 & |x|\leq \delta\\ \delta(|x|-\delta/2) & \text{otherwise} \end{array}\right.
\end{align*}
for some constant $\delta$, we take $G(x) = f(f^{-1}(1)
x)$. Then the norm $\|x\|_G$ looks like $\ell_2$ norm when $x$ is
flat, and looks like $\ell_1$ norm when $x$ is
sparse. Figure~\ref{fig:orlicz_norm_ball} shows the unit norm ball of
this kind of Orlicz norm.
%% Thus, this $\|\cdot\|_G$ can be one of
%% potential choices of the loss function $l(\cdot)$ in
%% Equation~\ref{eq:regression_prob}.

Our main result is a generic algorithm for solving any regression problem
Eqn.~\eqref{eq:regression_prob} with any loss function that is a
``nice'' Orlicz norm; see Section~\ref{sec:pre} for a formal
definition of ``nice'', and think of it as subquadratic for now.

\begin{figure}
  \centering
  \includegraphics[width=0.20\textwidth]{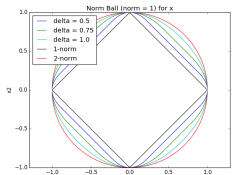}\\
  \vspace{-0.2in}
  \caption{\small Unit norm balls of Orlicz norm induced by normalized Huber functions with different $\delta$.}\label{fig:orlicz_norm_ball}
  \vspace{-0.3in}
\end{figure}

Our main result employs the concept of {\em subspace embeddings},
which is a powerful tool for solving numerical linear algebra
problems, and as such has many applications beyond regression. We say
that a subspace embedding matrix $S\in\mathbb{R}^{m\times n}$ embeds
the column space of $A\in\mathbb{R}^{n\times d}$ $(n>m)$ with $u$-norm
into a subspace with $v$-norm, if $\forall x\in\mathbb{R}^d,$ we have
$\|Ax\|_u/\alpha\leq \|SAx\|_v\leq \beta\|Ax\|_u$ where $\alpha\beta$
is called distortion (approximation). A long line of work studied
$\ell_2$ regression problem based on $\ell_2$ subspace embedding
techniques; see, e.g., \cite{cw09,cw13,nn13}. Furthermore, there are
works on $\ell_p$ regression problem based on $\ell_p$ subspace
embedding techniques (see, e.g.~\cite{sw11,mm13,cdmmmw13, wz13}), and
similarly for M-estimators \cite{cw15}.

Our overall results are composed of four parts:
\begin{enumerate*}
\item
%1.
We develop the first subspace embedding method for all ``nice'' Orlicz
norms. The embedding obtains a distortion factor polynomial in $d$,
which was recently shown necessary~\cite{ww18}.
\item
%2.
Using the above subspace embedding, we obtain the first approximation
algorithm for solving the linear regression problem with any ``nice''
Orlicz norm loss.
\item
%3.
As a further illustration of the power of the subspace embedding
method, we employ it towards improving on the best known result for
another problem: $\ell_p$ low rank approximation for $1\leq p<2$
from~\cite{swz17}, which is the ``$\ell_p$-version of PCA''.
\item
Finally, we complement our theoretical results with experimental evaluation of
our algorithms. Our experiments reveal that that the solution of
regression under the Orlicz norm induced by Huber loss is much better
than the solution given by regression under $\ell_1$ or $\ell_2$
norms, under natural noise distributions in practice. We also perform
experiments for Orlicz regression with different Orlicz functions $G$
and show their behavior under different noise settings, thus
exhibiting the flexibility of our framework.
\end{enumerate*}

%\aanote{needs to be rewritten -- i do not understand the next  sentence}
%% We give a generalization from $\ell_p$ techniques to the
%% Orlicz norms which is applicable in more settings. We view our main
%% contribution to be a unified framework for a wide class of general
%% norms (Orlicz norms), not improving upon previous results for the
%% $\ell_p$ norm.

To the best of our knowledge, our algorithms are the first low
distortion embedding and regression algorithms for general Orlicz
norm. For the problem of low rank approximation under $\ell_p$ norm,
$p\in[1,2)$, our algorithms achieve simultaneously the best
  approximation and the best running
  time. In contrast, all the previous algorithms achieve either the best
  approximation, or the best running time, but not both at the same time.

Our algorithms for subspace embedding and regression are simple, and
in particular are not iterative. In particular, for the subspace
embedding, the embedding matrix $S$ is generated independently of the
data. In the regression problem, we multiply the input with the
embedding matrix, and thus reduce to the $\ell_2$ regression problem,
for which we can use any of the known algorithm.

\vspace{-0.2in}
\paragraph{Technical discussion.} Next we highlight some of our techniques
used to obtain the theoretical results.

%% \aanote{to be incorporated into the discussion below, if necessary}\Peilin{Which of the following two sentences is better? ``we use ideas in STOC17 to xxx'' or ``we extend ideas in STOC17 to xxx''? ``we use'' may make reader feel that we just apply the idea directly then it works, but actually it is non-trivial to extend ``for-each'' case to ``for-all'' case, and I think this is a part of our novel contribution.} In this
%% paper, we use ideas in~\cite{anrw16} to generalize the work of
%% \cite{wz13} to show that it is possible to embed the subspace under a
%% ``nice'' Orlicz norm into a subspace with $\ell_2$ norm. Furthermore,
%% we show how to use this technique to solve
%% Eqn.~\eqref{eq:regression_prob}, again under a ``nice'' Orlicz
%% norm.
%% Finally, we show how our techniques work for the $\ell_p$ low
%% rank matrix approximation problem, improving over the work
%% of~\cite{swz17}, which has the best known theoretical guarantee for
%% $\ell_p$ low rank matrix approximation problem ($\ell_p$ version of
%% Principal Component Analysis, $1\leq p<2$).

\textit{Subspace embedding.} Our starting point is a technique
introduced in \cite{anrw16} for the Orlicz norms, which can be seen as
an embedding that has guarantees for a {\em fixed vector} only. In
contrast, our main challenge here is to obtain an embedding for {\em
  all vectors $x\in \R^n$} in a certain $d$-dimensional subspace.  Consider a random diagonal matrix
$D\in\mathbb{R}^{n\times n}$ with each diagonal entry is a
``generalized exponential'' random variable, i.e., drawn from a
distribution with cumulative distribution function $1-e^{-G(x)}.$
Then, for a fixed $x\in\mathbb{R}^d$, \cite{anrw16} show that
$\|D^{-1}Ax\|_\infty$ is not too small with high probability. We can
combine this statement together with a net argument and the dilation
bound on $\|D^{-1}Ax\|_G$, to argue that $\forall x\in\mathbb{R}^d,$
$\|D^{-1}Ax\|_\infty$ is not too small.

The other direction is more challenging --- to show that for a given
matrix $A\in\mathbb{R}^{n\times d}$, and any fixed $x\in\mathbb{R}^d,$
$\|D^{-1}Ax\|_G$ cannot be too large. Once we show this ``dilation
bound'', we combine it with the well-conditioned basis argument
(similar to \cite{ddhkm09}), and prove that $\forall
x\in\mathbb{R}^d,$ $\|D^{-1}Ax\|_G$ cannot be too large.  Overall, we
have that $\forall x\in\mathbb{R}^{d},$ $\|D^{-1}Ax\|_G\leq O(d^2\log
n) \cdot \|Ax\|_G,$ and $\|D^{-1}Ax\|_{\infty}\geq \Omega(1/(d\log n))
\cdot \|Ax\|_G.$ Since $\ell_2$ norm is sandwiched by $\|\cdot\|_G$
and $\ell_{\infty}$ norm, we have that $\forall x\in\mathbb{R}^d,
\Omega(1/(d\log n)) \cdot \|Ax\|_G\leq\|D^{-1}Ax\|_2\leq O(d^2\log n)
\cdot \|Ax\|_G.$ Then, the remaining part is to use standard
techniques~\cite{wz13,w14} to perform the $\ell_2$ subspace embedding for
the column space of $D^{-1}A.$ See
Theorem~\ref{thm:full_subspace_embedding} for
details.

The actual proof of the dilation bound is the most technically
intricate result.  Traditionally, since the $p^\text{th}$ power of the
$\ell_p$ norm is the sum of the $p^\text{th}$ power of all the
entries, it is easy to bound the expectation by using linearity of the
expectation. However it is impossible to apply this analysis to Orlicz
norm directly since Orlicz norm is not an "entrywise" norm. Instead,
we exploit a key observation that the Orlicz norm of vectors which are
on the unit ball can be seen as the sum of contribution of each
coordinate. Thus, we propose a novel analysis for any fixed vector by
analyzing the behavior of the normalized vector which is on the unit
Orlicz norm ball. To extend the dilation bound for a fixed vector to
all the vectors in a subspace, we generalize the definition of
$\ell_p$ norm well-conditioned basis to the Orlicz norm case, and then
show that the Auerbach basis provides a good basis for Orlicz norm. To
the best of our knowledge, this is the first time Auerbach basis are
used to analyze the dilation bound of an embedding method for a norm
distinct from an $\ell_p$ norm.
See Section~\ref{sec:subspace_for_orlicz} for details.

\textit{Regression with Orlicz norm.} Here, given a matrix
$A\in\mathbb{R}^{n\times d},$ a vector $b\in\mathbb{R}^n,$ the goal is
to solve Equation~\ref{eq:regression_prob} with Orlicz norm loss
function. We can now solve this problem directly using the subspace
embedding from above, in particular by applyingit to the column space
of $[A\ b].$ We obtain an $O(d^3\log^2 n)$ approximation ratio, which
we can further improve by observing that it actually suffices to have
the dilation bound on $\|D^{-1}Ax^*\|_G$ only for the optimal solution
$x^*$ (as opposed to for an arbitrary $x$). Using this observation, we
improve the approximation ratio to $O(d\log^2 n)$. See
Theorem~\ref{thm:regression_main} for details.
%% Generally, in the case of
%% over-constrained regression problem (i.e., when $n\gg d$), this
%% approximation ratio is acceptable.
We evaluate the algorithm's performance and show that it
performs well (even when $n$ is not much larger than $d$). See
Section~\ref{sec:expe}.

\textit{$\ell_p$ low rank matrix approximation.} The $\ell_p$ norm is
a special case of the Orlicz norm $\|\cdot\|_G$, where $G(x) = x^p.$ This
connection allows us to consider the following problem: given
$A\in\mathbb{R}^{n\times d},n\geq d\geq k\geq 1,$ find a rank-$k$
matrix $B\in\mathbb{R}^{n\times d}$ such that $\|A-B\|_p$ is
minimized. Here we consider the case of $1\leq p<2$ and $k=\omega(\log
n)$. The best known algorithm for this problem is from~\cite{swz17},
which uses the dense $p$-stable transform to achieves $k^2 \cdot
\mathrm{poly}(\log n)$ approximation ratio. It has the downside that
its runtime does not compare favorably to the golden standard of
runtime linear in the sparsity of the input. To improve the runtime,
one can apply the sparse $p$-stable transform and achieve input
sparsity runtime, but that comes at the cost of an $\Omega(k^6)$
factor loss in the approximation ratio.

Using the above techniques, we develop an algorithm with best of both
worlds: $k^2 \cdot \mathrm{poly}(\log n)$ approximation ratio and the
input sparsity running time at the same time. In particular, the main
inefficiency of the algorithm~\cite{swz17} is the relaxation from
$\ell_p$ norm to $\ell_2$ norm, which incurs a further
$\mathrm{poly}(k)$ approximation factor. In contrast, the embedding
based on exponential random variables embeds $\ell_p$ norm to $\ell_2$
norm directly, without further approximation loss. Our embedding also
comes with its own pitfalls --- as we need to deal with mixed norms
--- thus requiring a new analysis. See
Theorem~\ref{thm:lp_lowrank} for details.

\vspace{-0.1in}
%% The difficulty in the analysis is that since we embed the $\ell_p$ norm into $\ell_2$ norm, we need to handle the mixed norm. Thus the analysis has significant differences from~\cite{swz17}.

%\input{related.tex}

%% file: preli.tex
\section{Notations and preliminaries}\label{sec:pre}
\vspace{-0.1in}
%In this section, we mainly introduce some notations and definitions used in the paper.
%Define $\mathbb{R}$ to be the set of real numbers,
In this paper, we denote $\mathbb{R}_+$ to be the set of nonnegative reals.
%$\mathbb{Z}$ to be the set of all the integers.
Define $[n]=\{1,2,\cdots,n\}$. Given a matrix $A\in\mathbb{R}^{n\times d},$ $\forall i\in[n],j\in[d]$, $A^i$ and $A_j$ denotes the $i^{\text{th}}$ row and the $j^{\text{th}}$ column of $A$ respectively. $\mathrm{nnz}(A)$ denotes the number of nonzero entries of $A$. The column space of $A\in\mathbb{R}^{d}$ is $\{y\mid \exists x\in\mathbb{R}^d,y=Ax\}.$ $\forall p\not=2,$ $\|A\|_p \triangleq (\sum |A_{i,j}|^p)^{1/p}$, i.e. entrywise $p$-norm.
%i.e. $p$-th root of the sum of $p$-th powers of all the entries of matrix $A$.
$\|A\|_F$ defines the Frobenius norm of $A$, i.e. $(\sum A_{i,j}^2)^{1/2}$. $A^\dagger$ denotes the Moore-Penrose pseudoinverse of $A$. %Unless otherwise specified, all the vectors $x\in\mathbb{R}^n$ are $n$-dimensional column vectors, and $x_i$ denotes the $i^{\text{th}}$ entry of $x$. $\|x\|_p$ denotes $\ell_p$ norm of $x$, i.e. $\|x\|_p=(\sum_{i=1}^n x_i^p)^{1/p}.$
Given an invertible function $f(\cdot)$, let $f^{-1}(\cdot)$ be the inverse function of $f(\cdot)$.
 %i.e. for any $x$ in the domain of $f$, $f(f^{-1}(x))=f^{-1}(f(x))=x.$
 If $f(\cdot)$ is not invertible in $(-\infty,+\infty)$ but it is invertible in $[0,+\infty),$ then we denote $f^{-1}(\cdot)$ to be the inverse function of $f(\cdot)$ in domain $[0,+\infty)$. $\inf$ and $\sup$ denote the infimum and supremum respectively. $f'(x),f'_+(x),f'_-(x)$ denote the derivative, right derivative and left derivative of $f(x)$, respectively. Similarly, define $f''(x)$ for the second derivatives, and we define $f''_+(x)=\lim_{h\rightarrow 0^+}(f'(x+h)-f'_+(x))/h.$
In the following, we give the definition of Orlicz norm.

\begin{definition}[Orlicz norm]\label{def:orlicz_norm}
For any nonzero monotone nondecreasing convex function $G:\mathbb{R}_+\rightarrow \mathbb{R}_+$ with $G(0)=0$. Define Orlicz norm $\|\cdot \|_G$ as:
$
\forall n\in\mathbb{Z},n\geq 1, x\in\mathbb{R}^n,\|x\|_G=\inf\left\{\alpha>0\large\mid\sum_{i=1}^n G(|x_i|/\alpha)\leq 1\right\}.
$
\end{definition}
\vspace{-0.1in}
For any function $G_1(\cdot)$ which is valid to define an Orlicz norm, we can always ``simplify/normalize'' the function to get another function $G_2$ such that computing $\|\cdot\|_{G_1}$ is equivalent to computing $\|\cdot\|_{G_2}.$

\begin{fact}\label{fac:simpl}
Given a function $G_1:\mathbb{R}_+\rightarrow\mathbb{R}_+$ which can induce an Orlicz norm $\|\cdot\|_{G_1}$ (Definition~\ref{def:orlicz_norm}), define function $G_2:\mathbb{R}_+\rightarrow\mathbb{R}_+$ as the following:
$
G_2(x)=\left\{\begin{array}{ll}G_1(G_1^{-1}(1)x) & 0\leq x\leq 1\\ sx-(s-1) & x>1\end{array}\right.
$
where
$
s=\sup\left\{\left(G_2(y)-G_2(x)\right)/(y-x)\mid 0\leq x\leq y\leq 1\right\}.
$
Then $\|\cdot\|_{G_2}$ is a valid Orlicz norm. Furthermore, $\forall n\in\mathbb{Z},n\geq 1,x\in\mathbb{R}^n,$ we have $\|x\|_{G_1}=\|x\|_{G_2}/ G_1^{-1}(1).$
\end{fact}
\vspace{-0.1in}

Thus, without loss of generality, in this paper we consider the Orlicz norm induced by function $G$ which satisfies $G(1)=1$, and $G(x)$ is a linear function for $x>1$. In addition, we also require that $G(x)$ grows no faster than quadratically in $x$. Thus, we define the property $\mathcal{P}$ of a function $G:\mathbb{R}\rightarrow\mathbb{R}_+$ as the following:
1) $G$ is a nonzero monotone nondecreasing convex function in $[0,\infty)$;
2) $G(0)=0,G(1)=1,\forall x\in\mathbb{R},G(x)=G(-x)$;
3) $G(x)$ is a linear function for $x>1$, i.e. $\exists s>0,\forall x>1,G(x)=sx+(1-s)$;
4) $\exists \delta_G>0$ such that $G$ is twice differentiable on interval $(0,\delta_G)$. Furthermore, $G'_+(0)$ and $G''_+(0)$ exist, and either $G'_+(0)>0$ or $G''_+(0)>0$;
5) $\exists C_G>0,\forall 0<x<y,G(y)/G(x)\leq C_G(y/x)^2$.

The condition 1 is required to define an Orlicz norm. The conditions 2,3 are required because we can always do the simplification/normalization (see Fact~\ref{fac:simpl}). The condition 4 is required for the smoothness of $G$. The condition 5 is due to the subquadratic growth condition. Subquadratic growth condition is necessary for sketching $\sum_{i=1}^n G(x_i)$ with sketch size sub-polynomial in the dimension $n$, as shown by~\cite{bo10}. For example, if $G(x)= x^p$ for some $p>2,$ then $\|\cdot\|_G$ is the same as $\|\cdot\|_p$. It is necessary to take $\Omega(n^{1-2/p})$ space to sketch $\ell_p$ norm in $n$-dimensional space. Condition 5 is also necessary for $2$-concave property, \cite{ks85,ks89} shows that $\|\cdot\|_G$ can be embedded into $\ell_1$ space if and only if $G$ is $2$-concave. Although~\cite{s95} gives an explicit embedding to $\ell_1$, it cannot be computed efficiently.

There are many potential choices of $G(\cdot)$ which satisfies property $\mathcal{P}$, the following are some examples:
1) $G(x)= |x|^p$ for some $1\leq p\leq 2$. In this case $\|\cdot\|_G$ is exactly the $\ell_p$ norm $\|\cdot \|_p$;
2) $G(x)$ can be a normalized M-estimator function (see~\cite{z97}), i.e. define $f(x)$ to be one of the functions in Table~\ref{tab:M-estimators}.
%\begin{align*}
%1.~&f(x)\equiv\left\{\begin{array}{ll}x^2/2&|x|\leq c\\ c(|x|-c/2) & |x|>c\end{array}\right.&\text{Huber function}\\
%2.~&f(x)\equiv2(\sqrt{1+x^2/2}-1)&\text{$\ell_1-\ell_2$ function}\\
%3.~&f(x)\equiv c^2\left(|x|/c-\log(1+|x|/c)\right)&\text{``Fair'' function}\\
%\end{align*}
and let
$
G(x)=\left\{ \begin{array}{ll}f(f^{-1}(1)x)&|x|\leq 1\\G'_-(1)|x|-(G'_-(1)-1) & |x|>1\end{array}\right..
$

%\begin{table}[t]
%\small
%  \caption{\small Some of M-estimators.}
%  \label{tab:M-estimators}
%  \centering
%  \begin{tabular}{cc}

\begin{table}[t]
\vspace{-0.2in}
\caption{\small Some of M-estimators.}
\label{tab:M-estimators}
\begin{center}
\begin{small}
\begin{sc}
\begin{tabular}{lcc}
    \toprule
      \small{Huber} & \small{$\left\{\begin{array}{ll}x^2/2&|x|\leq c\\ c(|x|-c/2) & |x|>c\end{array}\right.$} \\ \hline
      \small{$\ell_1-\ell_2$} &  \small{$2(\sqrt{1+x^2/2}-1)$}   \\ \hline
      \small{``Fair"} & \small{$c^2\left(|x|/c-\log(1+|x|/c)\right)$}\\
\bottomrule
\end{tabular}
\end{sc}
\end{small}
\end{center}
\vspace{-0.2in}

\end{table}

%    \bottomrule
%  \end{tabular}
%  \vspace{-0.2in}
%\end{table}

%\subsection{Properties of $G(\cdot)$}\label{sec:pro}
%In this section, we introduce some useful properties of function $G$ with property $\mathcal{P}$. See Appendix for details of proofs of the following Lemmas.
The following presents some useful properties of function $G$ with property $\mathcal{P}.$ See Appendix for details of proofs of the following Lemmas.

\begin{lemma}\label{lem:func_grow_ratio}
Given a function $G(\cdot)$ with property $\mathcal{P}$, then $\forall 0\leq x\leq 1,$ $x^2/C_G\leq G(x)\leq x.$
\end{lemma}

\begin{lemma}\label{lem:sandwich}
Given a function $G(\cdot)$ with property $\mathcal{P}$, then
$
\forall x\in\mathbb{R}^n, \|x\|_2/\sqrt{C_G}\leq\|x\|_G\leq \|x\|_1.
$
\end{lemma}

\begin{lemma}\label{lem:fasterthanlinear}
Given a function $G(\cdot)$ with property $\mathcal{P}$, then $\forall 0<x<y,$ we have $y/x\leq G(y)/G(x)$.
\end{lemma}

%By Lemma~\ref{lem:fasterthanlinear} we can get the following Lemma.
\begin{lemma}\label{lem:decompose}
Given a function $G(\cdot)$ with property $\mathcal{P}$, there exist a constant $\alpha_G>0$ which may depend on $G$, such that $\forall 0\leq a,b,$ if $ab\leq 1,$ then $G(a)G(b)\leq \alpha_G G(ab).$
\end{lemma}

%% file: subspace.tex
\section{Subspace embedding for Orlicz norm using exponential random variables}\label{sec:subspace_for_orlicz}
\vspace{-0.1in} In this section, we develop the subspace embedding
under the Orlicz norms which are induced by functions $G$ with the
property $\mathcal{P}$. We first show how to embed the subspace with
$\|\cdot\|_G$ norm into a subspace with $\ell_2$ norm, and then we use
dimensionality reduction techniques for the $\ell_2$ norm. Overall, we
will prove Theorem~\ref{thm:full_subspace_embedding} stated at the end
of this section. Before discussing the details, we give formal
definitions of subspace embedding.

\begin{definition}[Subspace embedding for Orlicz norm]
Given a matrix $A\in\mathbb{R}^{n\times d}$, if $S\in\mathbb{R}^{m\times n}$ satisfies
$
\forall x\in\mathbb{R}^{d}, \|Ax\|_G/\alpha\leq \|SAx\|_v\leq \beta \|Ax\|_G
$
where $\alpha,\beta\geq 1,\|\cdot\|_v$ is a norm (can still be $\|\cdot\|_G$), then we say $S$ embeds the column space of $A$ with Orlicz norm into the column space of $SA$ with $v$-norm. The distortion is $\alpha\beta$.
\end{definition}
\vspace{-0.1in}

If the distortion and the $v$-norm are clear from the
context, we just say $S$ is a subspace embedding matrix for $A$.

\begin{definition}[Subspace embedding for $\ell_2$ norm]
Given a matrix $A\in\mathbb{R}^{n\times d}$, if $S\in\mathbb{R}^{m\times n}$ satisfies
$
\forall x\in\mathbb{R}^{d}, (1-\varepsilon)\|Ax\|_2^2\leq \|SAx\|_2^2\leq (1+\varepsilon) \|Ax\|_2^2,
$
then we say $S$ is a subspace embedding of column space of $A$.
\end{definition}
\vspace{-0.05in}

There are many choices of $\ell_2$ subspace embedding matrix $A$
satisfying the above definition. Examples are: random sign JL
matrix~\cite{a03,cw09}, fast JL matrix~\cite{ac09}, and sparse embedding
matrices~\cite{cw13,mm13,nn13}.

The main technical thrust is to embed $\|\cdot\|_G$ into $\ell_2$
norm. As the embedding matrix, we use $S=\Pi D^{-1}$ where $\Pi$ is
one of the above $\ell_2$ embedding matrices and $D$ is a diagonal
matrix of which diagonal entries are i.i.d. random variables draw from
the distribution with CDF $1-e^{-G(t)}.$ Equivalently, each entry on
the diagonal of $D$ is $G^{-1}(u),$ where $u$ is an i.i.d. sample from
the standard exponential distribution, i.e. CDF is $1-e^{-t}$. In
Section~\ref{sec:no_dilation}, we will prove that $\forall
x\in\mathbb{R}^d,$ $\|D^{-1}Ax\|_G$%\aanote{should this be $\|.\|_G$?
  %the same for the next sentence.}
  will not be too large. In
Section~\ref{sec:no_contraction}, we will show that $\forall
x\in\mathbb{R}^d,$ $\|D^{-1}Ax\|_\infty$ cannot be too small.
Then due to Lemma~\ref{lem:sandwich}, we know that $\|D^{-1}Ax\|_2$ is a good estimator to $\|Ax\|_G.$ In Section~\ref{sec:all_together}, we show how to put all the ingredients together.

\vspace{-0.1in}
\subsection{Dilation bound}\label{sec:no_dilation}
\vspace{-0.1in}
We construct a randomized linear map $f:\mathbb{R}^n \rightarrow \mathbb{R}^n$:
$ (x_1, x_2, ..., x_n) \xmapsto{f} \left( x_1/u_1, x_2/u_2, ..., x_n/u_n \right) $
where each $u_i$ is drawn from a distribution with CDF $1-e^{-G(t)}$. Notice that for proving the dilation bound, we do not need to assume $u_i$ are independent.

\begin{theorem}\label{thm:no_dilation_for_each}
Given $x\in\mathbb{R}^n$, let $\|\cdot\|_G$ be an Orlicz norm induced by function $G(\cdot)$ which has property $\mathcal{P}$, and let
$
f(x)= (x_1/u_1, x_2/u_2, ..., x_n/u_n),
$
where each $u_i$ is drawn from a distribution with CDF $1-e^{-G(t)}$. Then with probability at least $1-\delta-O(1/n^{19})$,
$
\|f(x)\|_G\leq O(\alpha_G\delta^{-1}\log(n))\|x\|_G,
$
where $\alpha_G$ is a constant may depend on function $G(\cdot)$.
\end{theorem}

\textbf{Proof sketch:} By taking union bound, we have $\forall i\in[n], u_i\geq G^{-1}(1/n^{20})$ with high probability.  Let $\alpha=\|x\|_G$. For $\gamma\geq 1,$ we want to upper bound the probability that $\|f(x)\|_G\geq \gamma \alpha.$ This is equivalent to upper bound the probability that $\|f(x)/(\gamma\alpha)\|_G\geq 1.$ Notice that $\Pr(\|f(x)/(\gamma\alpha)\|_G\geq 1)=\Pr(\sum G(x_i/\alpha\cdot 1/(\gamma u_i))\geq 1).$ By Markov inequality, it suffices to bound the expectation of $\sum G(x_i/\alpha\cdot 1/(\gamma u_i))$ conditioned on $u_i$ are not too small. By lemma~\ref{lem:decompose}, $\sum G(x_i/\alpha\cdot 1/(\gamma u_i))\leq \alpha_G/\gamma\cdot\sum G(x_i/\alpha)\cdot 1/G(u_i).$ Because $u_i$ is not too small, the conditional expectation of $1/G(u_i)$ is roughly $O(\log n).$ So the probability that $\|f(x)\|_G\geq \gamma \alpha$ is bounded by $O(\alpha_G\log n/\gamma),$ set $\gamma=O(\log n)\alpha_G/\delta$, we can complete the proof. See appendix for the details of the whole proof.

The final step is to use a well-conditioned basis; see details in
appendix.%\aanote{do we have appendix?}.
We then obtain the following
theorem.

\begin{theorem}\label{thm:no_dilation_for_all}
Let $G(\cdot)$ be a function which has property $\mathcal{P}.$ Given a matrix $A\in\mathbb{R}^{n\times m}$ with rank $d\leq n$, let $D\in\mathbb{R}^{n\times n}$ be a diagonal matrix of which each entry on the diagonal is drawn from a distribution with CDF $1-e^{-G(t)}.$ Then, with probability at least $0.99,$
$
\forall x\in\mathbb{R}^m, \|D^{-1}Ax\|_G\leq O(\alpha_Gd^2\log n)\|Ax\|_G,
$
where $\alpha_G\geq 1$ is a constant which only depends on $G(\cdot)$.
\end{theorem}

\vspace{-0.1in}
\subsection{Contraction bound}\label{sec:no_contraction}

As in Section~\ref{sec:no_dilation}, we construct a
randomized linear map $f:\mathbb{R}^n \rightarrow \mathbb{R}^n$: $
(x_1, x_2, ..., x_n) \xmapsto{f} \left( x_1/u_1, x_2/u_2, ..., x_n/u_n
\right) $ where each $u_i$ is an i.i.d. random variable drawn from a
distribution with CDF $1-e^{-G(t)}$. Notice that the difference from
proving the dilation bound is that we need $u_i$ to be independent
here. We use the following theorem:

\begin{theorem}[Lemma 3.1 of~\cite{anrw16}]\label{thm:no_contraction_for_each}
Given $x\in\mathbb{R}^n$, let $\|\cdot\|_G$ be an Orlicz norm induced by function $G(\cdot)$ which has property $\mathcal{P}$, and let
$
f(x)=\left( x_1/u_1, x_2/u_2, ..., x_n/u_n \right),
$
where each $u_i$ is an i.i.d random variable drawn from a distribution with CDF $1-e^{-G(t)}$. Then for $\alpha\geq 1,$ with probability at least $1-e^{-\alpha},$
$
\|f(x)\|_{\infty}\geq \|x\|_G/\alpha.
$
\end{theorem}

By combining the result with the net argument (see appendix), and
Theorems~\ref{thm:no_contraction_for_each},~\ref{thm:no_dilation_for_all},
we get the following:

\begin{theorem}\label{thm:no_contraction_for_all}
 $G(\cdot)$ is a function with property $\mathcal{P}.$ Given a matrix $A\in\mathbb{R}^{n\times m}$ with rank $d\leq n$, let $D\in\mathbb{R}^{n\times n}$ be a diagonal matrix of which each entry on the diagonal is an i.i.d. random variable drawn from the distribution with CDF $1-e^{-G(t)}.$ Then, with probability at least $0.98,$
$
\forall x\in\mathbb{R}^m, \Omega(1/(\alpha'_G d\log n))\|Ax\|_G\leq \|D^{-1}Ax\|_\infty,
$
where $\alpha'_G\geq 1$ is a constant which only depends on $G(\cdot)$.
\end{theorem}

\textbf{Proof sketch:} Set $\varepsilon=1/\mathrm{poly}(nd),$ we can build an $\varepsilon$-net (see Appendix) $N$ for the column space of $A$. By taking the union bound over all the net points, we have $\forall x\in N,$ $\|D^{-1}x\|_\infty$ is not too small. Due to Theorem~\ref{thm:no_dilation_for_all}, we have $\forall x$ in the column space of $A$, $\|D^{-1}x\|_G$ is not too large. Now, for any unit vector $y$ in the column space of $A$, we can find the closest point $x\in N,$ and $\|x-y\|_2\leq \varepsilon.$ Since $\|D^{-1}y\|_\infty\geq \|D^{-1}x\|_\infty-\|D^{-1}(y-x)\|_\infty,$ $\|D^{-1}x\|_\infty$ is not too small, and $\|D^{-1}(y-x)\|_\infty$ is not too large, we can get a lower bound for $\|D^{-1}y\|_\infty.$ See appendix for details.

\subsection{Putting it all together}\label{sec:all_together}

We now combine Theorem~\ref{thm:no_contraction_for_all},
Theorem~\ref{thm:no_dilation_for_all}, and  Lemma~\ref{lem:sandwich}, to get the following theorem.

\begin{theorem}\label{thm:sub_for_orlicz}
Let $G(\cdot)$ be a function which has property $\mathcal{P}.$ Given a matrix $A\in\mathbb{R}^{n\times m}$ with rank $d\leq n$, let $D\in\mathbb{R}^{n\times n}$ be a diagonal matrix of which each entry on the diagonal is an i.i.d. random variable drawn from the distribution with CDF $1-e^{-G(t)}.$ Then, with probability at least $0.98,$
$\forall x\in\mathbb{R}^m,
\Omega(1/(\alpha'_G d\log n))\|Ax\|_G \leq \|D^{-1}Ax\|_2
\leq O(\alpha''_Gd^2\log n)\|Ax\|_G,$
%\begin{align*}
%\forall x\in\mathbb{R}^m,&\\
%\Omega(1/(\alpha'_G d\log n))\|Ax\|_G &\leq \|D^{-1}Ax\|_2\\
%&\leq O(\alpha''_Gd^2\log n)\|Ax\|_G,
%\end{align*}
where $\alpha''_G,\alpha'_G\geq 1$ are two constants which only depend on $G(\cdot)$.
\end{theorem}

The above theorem successfully embeds $\|\cdot\|_G$ into $\ell_2$
space. We now use $\ell_2$ subspace embedding to reduce the
dimension. The following two theorems provide efficient $\ell_2$
subspace embeddings.

\begin{theorem}[~\cite{cw13}]\label{thm:sparse_embedding}
Given matrix $A\in\mathbb{R}^{n\times m}$ with rank $d$. Let $t=\Theta(d^2/\varepsilon^2)$, $S=\Phi Y\in\mathbb{R}^{t\times n},$ where $Y\in\mathbb{R}^{n\times n}$ is a diagonal matrix with each diagonal entry independently uniformly chosen to be $\pm 1$, $\Phi\in\mathbb{R}^{t\times n}$ is a binary matrix with $\Phi_{h(i),i}=1,\forall i\in[n],$ and remaining entries $0$. Here $h:[n]\rightarrow [t]$ is a random hashing function such that for each $i\in[n],$ $h(i)$ is uniformly distributed in $[t].$ Then with probability at least $0.99,$
$
\forall x\in\mathbb{R}^m, (1-\varepsilon)\|Ax\|_2^2\leq \|SAx\|_2^2\leq (1+\varepsilon)\|Ax\|_2^2.
$
Furthermore, $SA$ can be computed in $\mathrm{nnz}(A)$ time.
\end{theorem}

%Another choice for $\ell_2$ subspace embedding is JL matrix.
\begin{theorem}[See e.g.~\cite{w14}]\label{thm:dense_embedding}
Given matrix $A\in\mathbb{R}^{n\times m}$ with rank $d$. Let $t=\Theta(d/\varepsilon^2)$, $S\in\mathbb{R}^{t\times n}$ be a random matrix of i.i.d. standard Gaussian variables scaled by $1/\sqrt{t}.$ Then with probability at least $0.99,$
$
\forall x\in\mathbb{R}^m, (1-\varepsilon)\|Ax\|_2^2\leq \|SAx\|_2^2\leq (1+\varepsilon)\|Ax\|_2^2.
$
\end{theorem}

We conclude the full theorem for our subspace embedding:

\begin{theorem}\label{thm:full_subspace_embedding}
Let $G(\cdot)$ be a function which has property $\mathcal{P}.$ Given a matrix $A\in\mathbb{R}^{n\times d}$, $d\leq n$, let $D\in\mathbb{R}^{n\times n}$ be a diagonal matrix of which each entry on the diagonal is an i.i.d. random variable drawn from the distribution with CDF $1-e^{-G(t)}.$ Let $\Pi_1\in\mathbb{R}^{t_1\times n}$ be a sparse embedding matrix (see Theorem~\ref{thm:sparse_embedding}) and let $\Pi_2\in\mathbb{R}^{t_2\times t_1}$ be a random Gaussian matrix (see Theorem~\ref{thm:dense_embedding}) where $t_1=\Omega(d^2),t_2=\Omega(d).$ Then, with probability at least $0.9,$
%\begin{align*}
%\forall x\in\mathbb{R}^d,& \\
%\Omega(1/(\alpha'_G d\log n))\|Ax\|_G &\leq \|\Pi_2\Pi_1D^{-1}Ax\|_2\\
%&\leq O(\alpha''_Gd^2\log n)\|Ax\|_G,
%\end{align*}
$\forall x\in\mathbb{R}^d,
\Omega(1/(\alpha'_G d\log n))\|Ax\|_G \leq \|\Pi_2\Pi_1D^{-1}Ax\|_2
\leq O(\alpha''_Gd^2\log n)\|Ax\|_G,$
where $\alpha''_G,\alpha'_G\geq 1$ are two constants which only depend on $G(\cdot)$. Furthermore, $\Pi_2\Pi_1D^{-1}A$ can be computed in $\mathrm{nnz}(A)+\mathrm{poly}(d)$ time.
\end{theorem}

%% file: applications.tex
\vspace{-0.2in}
\section{Applications}\label{sec:app}
\vspace{-0.1in}
In this section, we discuss regression problem with Orlicz norm error
measure, and low rank approximation problem with $\ell_p$ norm, which
is a special case of the Orlicz norms.
\vspace{-0.1in}
\subsection{Linear regression under Orlicz norm}
\vspace{-0.1in}
We give the definition of regression problem as follow.

\begin{definition}
Function $G(\cdot)$ has property $\mathcal{P}$. Given $A\in\mathbb{R}^{n\times d},b\in\mathbb{R}^n,$ the goal is to solve the following minimization problem $\min_{x\in\mathbb{R}^d} \|Ax-b\|_G.$
\end{definition}
 \begin{algorithm}[h!]\caption{\small Linear regression with Orlicz norm $\|\cdot\|_G$}\label{alg:Orlicz_regression}
 \small
\begin{algorithmic}[1]
\small
\STATE \textbf{Input:} $A\in\mathbb{R}^{n\times d},b\in\mathbb{R}^n.$
\STATE \textbf{Output:} $\hat{x}$.
\STATE Let $t_1=\Theta(d^2),t_2=\Theta(d).$
\STATE Let $\Pi_1\in\mathbb{R}^{t_1\times n}$ be a random sparse embedding matrix, $\Pi_2\in\mathbb{R}^{t_2\times t_1}$ be a random gaussian matrix, and $D\in\mathbb{R}^{n\times n}$ be a random diagonal matrix with each diagonal entry independently drawn from distribution whose CDF is $1-e^{-G(t)}$. (See Theorem~\ref{thm:full_subspace_embedding}.)
\STATE Compute $\hat{x}=(\Pi_2\Pi_1D^{-1}A)^\dagger \Pi_2\Pi_1D^{-1}b.$
\end{algorithmic}
\end{algorithm}

\begin{theorem}\label{thm:regression_main}
Let $G(\cdot)$ have property $\mathcal{P}$. Given $A\in\mathbb{R}^{n\times d},b\in\mathbb{R}^n,$ Algorithm~\ref{alg:Orlicz_regression} can output a solution $\hat{x}\in\mathbb{R}^d$ such that with probability at least $0.8$,
$
\|A\hat{x}-b\|_G\leq O(\beta_G d\log^2 n)\min_{x\in\mathbb{R}^d}\|Ax-b\|_G,
$
where $\beta_G$ is a constant which may depend on $G(\cdot).$ In addition, the running time of Algorithm~\ref{alg:Orlicz_regression} is $\mathrm{nnz}(A)+\mathrm{poly}(d).$
\end{theorem}

\textbf{Proof sketch:} Let $S=\Pi_2\Pi_1D^{-1}$ be the subspace embedding for column space of $[A\ b]$. Let $x^*=\arg\min_{x\in\mathbb{R}^d}\|Ax-b\|_G.$ Due to Theorem~\ref{thm:full_subspace_embedding}, $\|A\hat{x}-b\|_G$ is bounded by $O(d\log n)\|S(A\hat{x}-b)\|_2\leq O(d\log n)\|S(Ax^*-b)\|_2\leq O(d\log n)\|D^{-1}(Ax^*-b)\|_2.$ Due to Theorem~\ref{thm:no_dilation_for_each}, $\|D^{-1}(Ax^*-b)\|_2\leq O(1)\|D^{-1}(Ax^*-b)\|_G\leq O(\log n)\|Ax^*-b\|_G.$

\vspace{-0.1in}
\subsection{Regression with combined loss function}
\vspace{-0.1in}
In this section, we want to point out that our technique can be used on solving regression problem with more general cost function. Recall that the goal is to solve the minimization problem $\min_{x\in \mathbb{R}^d}\Vert Ax-b \Vert_G$. Now, we consider there are multiple goals, and we want to minimize a linear combination of the costs. Now we give the definition of regression problem with combined cost function.

\begin{definition}
Suppose function $G_1(\cdot), G_2(\cdot),..., G_k(\cdot)$ satisfies property $\mathcal{P}$. Given $A_1\in \mathbb{R}^{n_1\times d},A_2\in \mathbb{R}^{n_2 \times d},...,A_k\in \mathbb{R}^{n_k\times d}, b_1 \in \mathbb{R}^{n_1},b_2\in \mathbb{R}^{n_2},...,b_k\in \mathbb{R}^{n_k}$, the goal is to solve the following minimization problem $\min_{x\in\mathbb{R}^d} \sum_{i=1}^k \Vert A_i x-b_i\Vert_{G_i}$.
\end{definition}
\vspace{-0.1in}
The idea of solving this problem is that we can embed every term into $l_1$ space, and then merge them into one term. By the standard technique, there is a way to embed $l_2$ space to $l_1$ space. We show the embedding as below. For the completeness, we put the proof of this lemma to the appendix.

\begin{lemma}
\label{lem:l2tol1}
Let $Q\in \mathbb{R}^{t\times n}$ be a random matrix with each entry drawn uniformly from i.i.d. $\mathcal{N}(0, 1)$ Gaussian distribution. Let $B=(\sqrt{\pi/2}/t)\cdot Q$. If $t=\Omega(\epsilon^{-2}n\log(n\epsilon^{-1}))$, then with probability at least 0.98, $\forall x \in \mathbb{R}^n, \Vert Bx\Vert_1 \in ((1 - \epsilon)\Vert x \Vert_2,(1 + \epsilon)\Vert x \Vert_2)$.
\end{lemma}

 \begin{algorithm}[h!]\caption{\small Linear regression with combined loss functions}\label{alg:general_regression}
 \small
\begin{algorithmic}[1]
\small
\STATE \textbf{Input:} $A_1\in \mathbb{R}^{n_1\times d},A_2\in \mathbb{R}^{n_2 \times d},...,A_k\in \mathbb{R}^{n_k\times d}, b_1 \in \mathbb{R}^{n_1},b_2\in \mathbb{R}^{n_2},...,b_k\in \mathbb{R}^{n_k}$
\STATE \textbf{Output:} $\hat{x}$.
\STATE Let $t_1 = \Theta(d^2), t_2 = \Theta(d), t_3 = \Theta(t_2 \log(t_2))$.
\STATE Let $\Pi_1^{(1)}\in \mathbb{R}^{t_1\times n_1},\cdots \Pi_1^{(k)}\in \mathbb{R}^{t_1\times n_k}$ be $k$ random sparse embedding matrices, $\Pi_2^{(1)},\cdots,\Pi_2^{(k)}\in \mathbb{R}^{t_2\times t_1}$ be $k$ random Gaussian matrices, and $D^{(1)}\in\mathbb{R}^{n_1\times n_1},\cdots,D^{(k)}\in\mathbb{R}^{n_k\times n_k}$ be $k$ random diagonal matrices where each diagonal entry of $D^{(i)}$ is independently drawn from distribution whose CDF is $1-e^{-G_i(t)}$. (See Theorem~\ref{thm:full_subspace_embedding}.) Let $Q^{(1)},\cdots, Q^{(k)} \in \mathbb{R}^{t_3\times t_2}$ be random matrices with each entry drawn uniformly from i.i.d. $\mathcal{N}(0, 1)$ Gaussian distribution. $\forall i\in[k],$ let $B^{(i)}=(\sqrt{\pi/2}/t_3)\cdot Q^{(i)}$ (see Lemma~\ref{lem:l2tol1}.) Let $B\in \mathbb{R}^{kt_3\times kt_2},\Pi_2\in\mathbb{R}^{kt_2\times kt_1},\Pi_1\in\mathbb{R}^{kt_1\times \sum_{j=1}^k n_j},D\in\mathbb{R}^{\sum_{j=1}^k n_j\times \sum_{j=1}^k n_j}$ be four block diagonal matrices such that $\forall i\in[k],$ the $i^{\text{th}}$ block of $B,\Pi_2,\Pi_1,D$ is $B^{(i)},\Pi_2^{(i)},\Pi_1^{(i)},D^{(i)}$ respectively.
\STATE Let $A = [A_1^{\top}, A_2^{\top},..., A_k^{\top}]^{\top}, b = [b_1^{\top}, b_2^{\top}, ..., b_k^{\top}]^{\top}$ and $S = B\Pi_2 \Pi_1 D^{-1}$.
\STATE Use classical method of solving $l_1$ regression to get $\hat{x} = \argmin_{x\in \mathbb{R}^d} \Vert S(Ax-b)\Vert_1$.
\end{algorithmic}
\end{algorithm}

\begin{theorem}
\label{general_regression}
Let $k>0$ be a constant, and $G_1(\cdot), G_2(\cdot),..., G_k(\cdot)$ satisfy property $\mathcal{P}$. Given $A_1\in \mathbb{R}^{n_1\times d},A_2\in \mathbb{R}^{n_2 \times d},...,A_k\in \mathbb{R}^{n_k\times d}, b_1 \in \mathbb{R}^{n_1},b_2\in \mathbb{R}^{n_2},...,b_k\in \mathbb{R}^{n_k}$, Algorithm \ref{alg:general_regression} can output a solution $\hat{x}\in \mathbb{R}^d$ such that with probability at least 0.7, $\sum_{i=1}^k \Vert A_i\hat{x}-b_i\Vert_{G_i} \leq O(\beta'_G d\log^2 n)\min_{x\in\mathbb{R}^d} \sum_{i=1}^k \Vert A_i x-b_i\Vert_{G_i}$, where $\beta'_G$ is a constant which may depend on $G_1(\cdot), G_2(\cdot),..., G_k(\cdot)$. In addition, the running time of Algorithm \ref{alg:general_regression} is $\sum_{i=1}^k \mathrm{nnz}(A_i) + \mathrm{poly}(d)$.
\end{theorem}

\textbf{Proof Sketch:}
Let $A = [A_1^{\top}, A_2^{\top},..., A_k^{\top}]^{\top}, b = [b_1^{\top}, b_2^{\top}, ..., b_k^{\top}]^{\top}$, and $S=B\Pi_2 \Pi_1 D^{-1}$ be the subspace embedding for column space of $[A\ b]$. Let $S_i=B^{(i)}\Pi_2^{(i)}\Pi_1^{(i)}(D^{(i)})^{-1}.$
Notice that $\forall x, \Vert S(Ax-b)\Vert_1 = \sum_{i=1}^k \Vert S_i(A_i x-b_i)\Vert_1$. Let $x^*=\argmin_{x\in\mathbb{R}^d} \sum_{i=1}^k \Vert A_i x-b_i\Vert_{G_i}$. Due to Theorem \ref{thm:full_subspace_embedding} and Lemma \ref{lem:l2tol1}, $\sum_{i=1}^k \Vert A_i\hat{x}-b_i\Vert_{G_i}$ is bounded by $O(d\log n)\sum_{i=1}^k \Vert S_i(A_i\hat{x}-b_i)\Vert_1 = O(d\log n)\Vert S(A\hat{x}-b)\Vert_1 \leq O(d\log n) \Vert S(Ax^*-b)\Vert_1 = \sum_{i=1}^k \Vert S_i(A_i x^*-b_i)\Vert_1$.  Due to Theorem \ref{thm:full_subspace_embedding}, $\sum_{i=1}^k \Vert S_i(A_i x^*-b_i)\Vert_1 \leq O(\log n)\sum_{i=1}^k \Vert A_i x^*-b_i\Vert_{G_i}$.

One application of the above Theorem is to solve the LASSO (Least Absolute Shrinkage Sector Operator) regression. In LASSO regression problem, the goal is to minimize $\Vert Ax-b\Vert_2^2 + \lambda \Vert x\Vert_1$, where $\lambda$ is a parameter of regularizer. It is easy to show that it is equivalent to  minimize $\Vert Ax-b\Vert_2 + \lambda' \Vert x\Vert_1$ for some other parameter $\lambda'.$ When we look at $\Vert Ax-b\Vert_2 + \lambda' \Vert x\Vert_1,$ we can set $A_1=A,b_1=b,A_2=\lambda'I,b_2=0,G_1(\cdot)\equiv x^2,G_2(\cdot)\equiv x,$ then we are able to apply Theorem~\ref{general_regression} to give a good approximation. The merit of our algorithm is that it is very simple, and can be computed very fast.

\vspace{-0.1in}
\subsection{$\ell_p$ norm low rank approximation using exponential random variables}
\vspace{-0.1in}
We discuss a special case of Orlicz norm $\|\cdot\|_G$, $\ell_p$ norm, i.e. $G(x)\equiv x^p$ for $p\in[1,2].$ When rank parameter $k$ is $\omega(\log n+\log d),$ by using exponential random variables, we can significantly improve the approximation ratio of input sparsity time algorithms shown by~\cite{swz17}. The high level ideas combine the results of~\cite{wz13,swz17} and the dilation bound in Section~\ref{sec:subspace_for_orlicz}. We define the problem in the following. See Appendix for the proof of Theorem~\ref{thm:lp_lowrank}.

\begin{definition}
Let $p\in[1,2].$ Given $A\in\mathbb{R}^{n\times d},n\geq d,k\in \mathbb{Z},1\leq k\leq \min(n,d),$ the goal is to solve the following minimization problem:
$
\min_{U\in\mathbb{R}^{n\times k},V\in\mathbb{R}^{k\times d}} \|UV-A\|_p^p.
$
\end{definition}

\vspace{-0.1in}
 \begin{algorithm}[h!]\caption{\small $\ell_p$ norm low rank approximation using exponential random variables.}\label{alg:lp_lowrank}
 \small
\begin{algorithmic}[1]
\STATE \textbf{Input:} $A\in\mathbb{R}^{n\times d},k\in \mathbb{Z},\min(n,d)\geq k\geq 1.$
\STATE \textbf{Output:} $\hat{U}\in\mathbb{R}^{n\times k},\hat{V}\in\mathbb{R}^{k\times d}$.
\STATE Let $t_1=\Theta(k^2),t_2=\Theta(k),t_3=\Theta(k\log k).$
\STATE Let $\Pi_1,S_1\in\mathbb{R}^{t_1\times n}$ be two random sparse embedding matrices, $\Pi_2,S_2\in\mathbb{R}^{t_2\times t_1}$ be two random gaussian matrices, and $D_1,D_2\in\mathbb{R}^{n\times n}$ be two random diagonal matrices with each diagonal entry independently drawn from distribution whose CDF is $1-e^{-t^p}$. (See Theorem~\ref{thm:full_subspace_embedding}.)
\STATE Let $T_2,R\in\mathbb{R}^{d\times t_3}$ be two random matrix, with i.i.d. entries drawn from standard $p$-stable distribution.
\STATE Let $S=S_2S_1D_1^{-1},T_1=\Pi_2\Pi_1D_2^{-1}$.
\STATE Solve $\hat{X},\hat{Y}=\arg\min_{X\in\mathbb{R}^{t_2\times k},Y\in\mathbb{R}^{k\times t_3}}\|T_1ARXYSAT_2-T_1AT_2\|_F^2.$
\STATE $\hat{U}=AR\hat{X},\hat{V}=\hat{Y}SA.$
\end{algorithmic}
\end{algorithm}
\vspace{-0.1in}

\begin{theorem}\label{thm:lp_lowrank}
Let $1\leq p\leq 2.$ Given $A\in\mathbb{R}^{n\times d},n\geq d,k\in \mathbb{Z},1\leq k\leq \min(n,d),$ with probability at least $2/3$, $\hat{U},\hat{V}$ outputted by Algorithm~\ref{alg:lp_lowrank} satisfies:
$
\|\hat{U}\hat{V}-A\|_p^p\leq \alpha\min_{U\in\mathbb{R}^{n\times k},V\in\mathbb{R}^{k\times d}} \|UV-A\|_p^p,
$
where $\alpha=O(\min((k\log k)^{4-p}\log^{2p+2}n, (k\log k)^{4-2p}\log^{4+p} n)).$
In addition, the running time of Algorithm~\ref{alg:lp_lowrank} is $\mathrm{nnz}(A)+(n+d)\mathrm{poly}(k).$
\end{theorem}

%% file: experiments.tex
\vspace{-0.2in}
\section{Experiments}\label{sec:expe}
\vspace{-0.1in}
Implementation setups can be seen in appendix.
\vspace{-0.1in}
\subsection{Orlicz Norm Linear Regression}
\vspace{-0.1in}
In this section, we show that our algorithm i) has reasonable and predictable performance under different scenarios and ii) is flexible, general and easy to use. We perform 3 sets of experiments. The first is to compare its performance with the standard $\ell_1$ and $\ell_2$ regression under different noise assumptions and dimensions of the regression problem; the second is to compare the performance of Orlicz regression with different $G$ under different noise assumptions; the third is to experiment with Orlicz function $G$ that is different from standard $\ell_p$ and Huber function. We evaluate the performance of our Orlicz norm linear regression algorithm on simulated data.

\vspace{-0.1in}
 \textbf{Comparison with $\ell_1$ and $\ell_2$ regression} We would like to see whether Orlicz norm linear regression leads to expected performance relative to $\ell_1$ and $\ell_2$ regression. We choose our Orlicz norm $\|\cdot\|_G$ to be induced by the normalized Huber function where the Huber function is defined as %\\
% $G(x)=\left\{\begin{array}{ll}
% (1/\delta+\delta/2)^2x^2/2 & |(1/\delta+\delta/2)x|\leq \delta\\
% \delta((1/\delta+\delta/2)|x|-\delta/2) & \text{o.w.}\end{array}\right.,$
 $f(x)=\left\{\begin{array}{ll}
 x^2/2 & |x| \leq \delta\\
 \delta \cdot ( |x| -  \delta/2) & \text{o.w.}\end{array}\right..$
 We chose the parameter $\delta$ to be $0.75$. Intuitively, it is between $\ell_1$ and $\ell_2$ norm (see Figure~\ref{fig:orlicz_norm_ball}). In all the simulations, we generate matrix $A\in\mathbb{R}^{n\times d}$, ground truth $x^*\in\mathbb{R}^d,$ and $b$ to be $Ax^*$ plus some particular noise. We evaluate the performance of each algorithm by the $\ell_2$ distances between the output $x$ and the ground truth $x^*$. In terms of algorithm details, since $n, d$ are not too large in our simulation, we did not apply the $\ell_2$ subspace embedding to reduce the dimension; we only use reciprocal exponential random diagonal embedding matrix to embed $\|\cdot\|_G$ to $\ell_2$ norm (see Theorem~\ref{thm:sub_for_orlicz})\footnote{We use MATLAB's \textit{linprog} to solve $\ell_1$ regression.}.

We experiment with two $n, d$ combinations, i) $n = 200, d = 10$ ii) $n=100, d=75$, and 3 noise setting with i) Gaussian noise ii) sparse noise and iii) mixed noise (addition of i) and ii)), altogether $2 \times 3 = 6$ setting. The detail of data simulation can be seen in appendix.
 For each experiment we repeat $50$ times and compute the mean. The
 results are shown in Table ~\ref{tab:control_experiments}. Orlicz
 norm regression has better performance than $\ell_1$ and $\ell_2$
 when the noise is mixed. When the noise is Gaussian or sparse, Orlicz
 norm regression works better than $\ell_1$ and $\ell_2$ respectively.
 We did not experiment with Huber loss regression, since if we rescale
 the data and make it small/large in absolute values, the Huber
 regression will degenerate into respectively $\ell_2$/$\ell_1$ regression (see Introduction). See appendix for results on approximation ratio.

\begin{table}
\vspace{-0.1in}
\small
\caption{Comparisons of different regressions in different noise and dimension settings; each entry is the error of $\ell_1, \ell_2$, Orlicz norm regression. As expected, $\ell_2$/$\ell_1$ regression lead to best performance under Sparse/Gaussian noise setting, and the performance of Orlicz norm regression lies in between.}
\label{tab:control_experiments}
\centering
\begin{tabular}{cccc}
\toprule
& Gaussian & Sparse & Mixed\\
\midrule
balance & \tiny{211.2/\textbf{194.5}/197.3} & \tiny{\textbf{25.3}/30.7/30.0} & \tiny{37.9/37.8/\textbf{37.5}}\\
overconstraint & \tiny{25.3/\textbf{20.0}/24.9} & \tiny{\textbf{2e-9}/1.6/1.5} &  \tiny{8.7/7.6/\textbf{7.5}}\\
\bottomrule
\end{tabular}
\vspace{-0.2in}
\end{table}

\textbf{Choice of $\delta$ for $G$ as a normalized Huber function} We compare the performance of Orlicz norm regression induced by $G$ as normalized Huber loss function with different $\delta$ under different noise assumptions. We fix $n = 500, d = 30$ and generate $A$ and $x$ as in the first set of experiments (see appendix). The noise is a mixture of $N(0,5)$ Gaussian noise and sparse noise on $1\%$ entries with different scale of uniform noise from $[-s\|Ax^*\|_2, s\|Ax^*\|_2]$, where scale $s$ is chosen from [0, 0.5, 1, 2]. Under each noise assumptions with different scale $s$, we compare the performance of Orlicz norm regression induced by $G$ with $\delta$ from [0.05, 0.1, 0.2, 0.4, 1, 2]. We repeat each experiment $50$ times and report the mean of the $\ell_2$ distance between output $x$ and the ground truth $x^*.$ The result is shown in Figure ~\ref{fig:best_delta}. When the scale is 0/2, the noise is almost Gaussian/sparse and we expect $\ell_2$/$\ell_1$ norm and thus large/small $\delta$ to perform the best; anything scale lying in between these extremes will have an optimal $\delta$ in between. We observe the expected trend: as $s$ increases, the performance is optimal with smaller $\delta$.

\begin{figure}[b]
\vspace{-0.2in}
  \centering
  \includegraphics[width=0.25\textwidth]{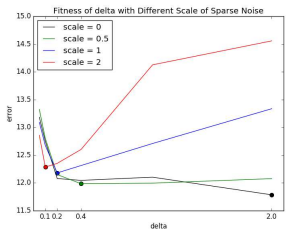}\\
  \vspace{-0.15in}
  \caption{\small Performance of Orlicz regression with $G$ induced by different $\delta$ under different scale of sparse noise. The larger the sparse noise, the smaller the $\delta$ that leads to the best performance, which makes the norm closer to $\ell_1$}\label{fig:best_delta}
  \vspace{-0.1in}
\end{figure}

\textbf{Beyond Huber function - A General Framework} We explore a variant Orlicz function $G$ and evaluate it under a particular setting; the evaluation criteria is the same as the first set of the experiments. The $G$ is of the same form aforementioned, except that it now grows at the order of $x^{1.5}$ when $x$ is small. We denote it by $G_{\ell_{1.5}}$, which is the normalization of function $f$, and $f$ is defined as: $f(x)=\left\{\begin{array}{ll}
x^{1.5}/1.5 & x \leq \delta\\
\delta^{0.5} \cdot( |x| - \delta/3)  & \text{o.w.}\end{array}\right.$. %($f$ does not satisfy property $P$, but Theorem \ref{thm:regression_main} still holds.)
We generate a $500 \times 30$ matrix $A$ and the ground truth vector $x^*$ in the same way as before, and then add $N(0,5)$ Gaussian noises and $1$ sparse outlier with scale $s = 100$. We find that the modified $G_{\ell_{1.5}}$ under this settings outperforms $\ell_1, \ell_2, \ell_{1.5}, G_{\delta=0.25}, G_{\delta=0.75}$ regression by a significant amount where $G_{\delta=0.25},G_{\delta=0.75}$ are Orlicz norm induced by regular normalized Huber function with $\delta=0.25,0.75$ respectively. The results are shown in Table ~\ref{tab:l1.5}. This experiment demonstrates that our algorithm is i) flexible enough to combine the advantage of norm functions, ii) general for any function that satisfies the nice property, and iii) easy to experiment with different settings, as long as we can compute $G$ and $G^{-1}.$

\begin{table}
\vspace{-0.1in}
\tiny
\caption{Orlicz regression with different choices of $G$, mean of the $\ell_{2}$ distances between the output and the ground truth in $50$ repetitions of experiments.}
\label{tab:l1.5}
\centering
\begin{tabular}{cccccc}
%\toprule
%& median & mean \\
%\midrule
%$l_1$ & 14.5 & 17.0\\
%$l_{1.5}$ & 31.7 & 45.0\\
%$l_2$ & 752.8 & 909.8\\
%$G_{\delta=0.25}$ & 45.0 & 60.2\\
%$G_{\delta=0.75}$ & 322.7 & 405.7\\
%$G_{\ell_{1.5}}$ & \textbf{13.45} & \textbf{14.7}\\
%\bottomrule

\toprule
\tiny{$\ell_1$} & \tiny{$\ell_{1.5}$} & \tiny{$\ell_2$}&\tiny{$G_{\delta=0.25}$} & \tiny{$G_{\delta=0.75}$} & \tiny{$G_{\ell_{1.5}}$} \\
 \midrule
 \tiny{$17.0$}&\tiny{$45.0$}&\tiny{$909.8$}&\tiny{$60.2$}&\tiny{$405.7$}&\tiny{$\mathbf{14.7}$}\\
\bottomrule
\end{tabular}
\vspace{-0.25in}
\end{table}

\vspace{-0.1in}
\subsection{$\ell_1$ low rank matrix approximation}
\vspace{-0.1in}
In this section, we evaluate the performance of the $\ell_1$ low rank matrix approximation algorithm. We mainly compare the $\ell_1$ norm error of our algorithm with the error of ~\cite{swz17} and standard PCA. Inputs are a matrix $A\in\mathbb{R}^{n\times d}$ and a rank parameter $k$; the goal is to output a rank $k$ matrix $B$ such that $\|A-B\|_1$  is as small as possible. The details of implementations are in the appendix. For each input, we run the algorithm $50$ times and pick the best solution.

\textbf{Datasets.} We first run experiment on synthetic data: we
randomly choose two matrices $U\in\mathbb{R}^{2000\times
5},V\in\mathbb{R}^{5\times 2000}$ with each entry drawn uniformly from
$(0,1)$ Then we randomly choose $100$ entries of $UV$, and add random
outliers uniformly drawn from $(-100,100)$ on those entries, thus we
can get a matrix $A\in\mathbb{R}^{2000\times 2000}$. In our
experiment, $\|A\|_1$ is about $5.0\times 10^6.$
Then, we run experiments on real datasets \textit{diabetes} and \textit{glass} in UCI repository\cite{bache2013uci}. The data matrix of \textit{diabetes} has size $768\times 8$, and the data matrix of \textit{glass} has size $214\times 9$. For each data matrix, we randomly add outliers on $1\%$ number of entries.

For each dataset, we evaluate the $\|A-B\|_1.$ The result for the experiment on synthetic data is shown in Table~\ref{tab:syn_data}, and the results for \textit{diabetes} and \textit{glass} are shown in Figure~\ref{fig:ucidata}. The running time of algorithm in~\cite{swz17} on \textit{diabetes} and on \textit{glass} are $5.69$ and $11.97$ seconds respectively, with ours being $3.18$ and $3.74$ seconds respectively. We also find that our algorithm consistently outperforms the other two alternatives (the y-coordinates are at log scale with base 10).

\begin{figure}[h!]
\vspace{-0.1in}
\noindent\begin{tabular}{cc}
  \includegraphics[width=0.22\textwidth]{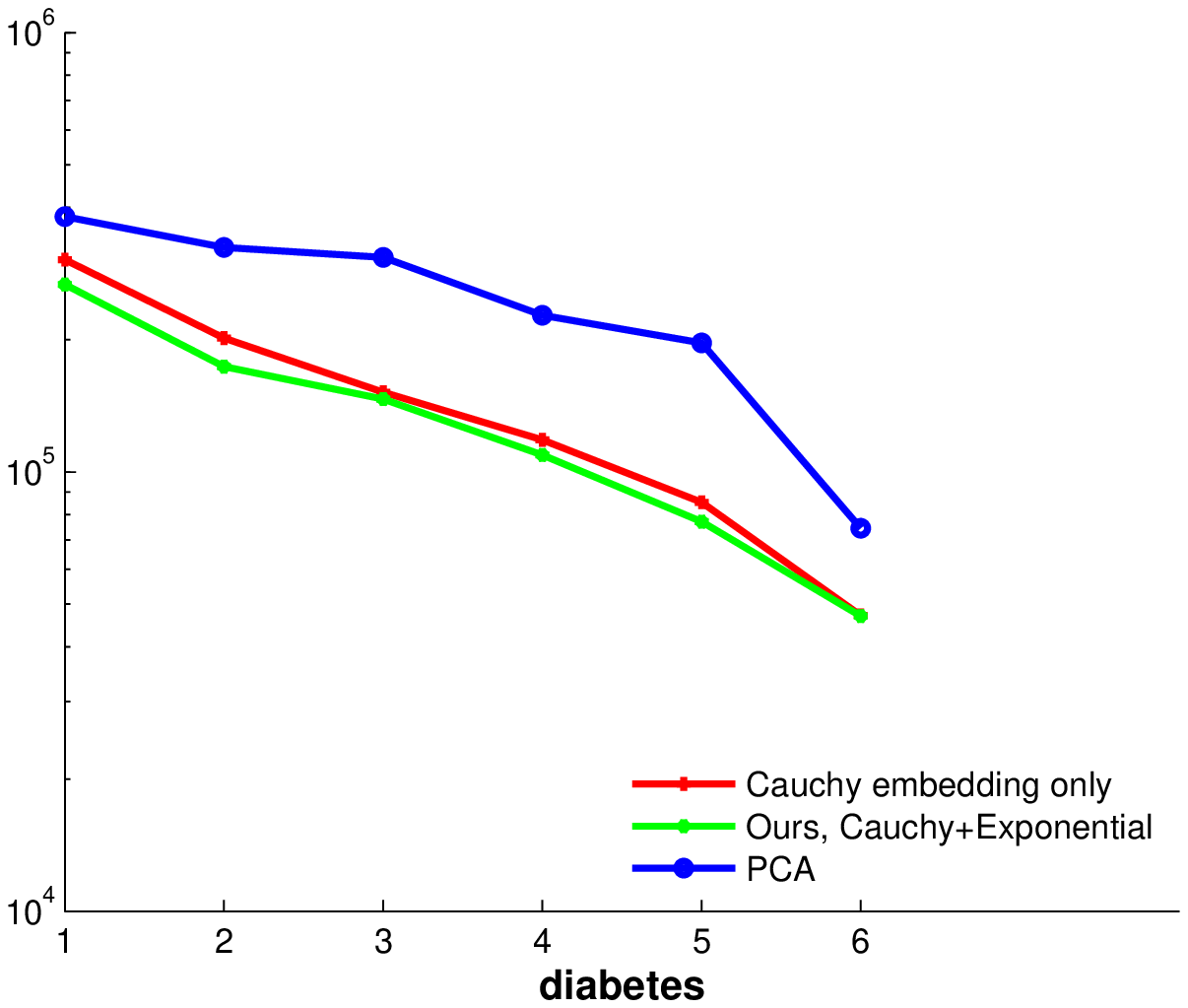}&
  \includegraphics[width=0.22\textwidth]{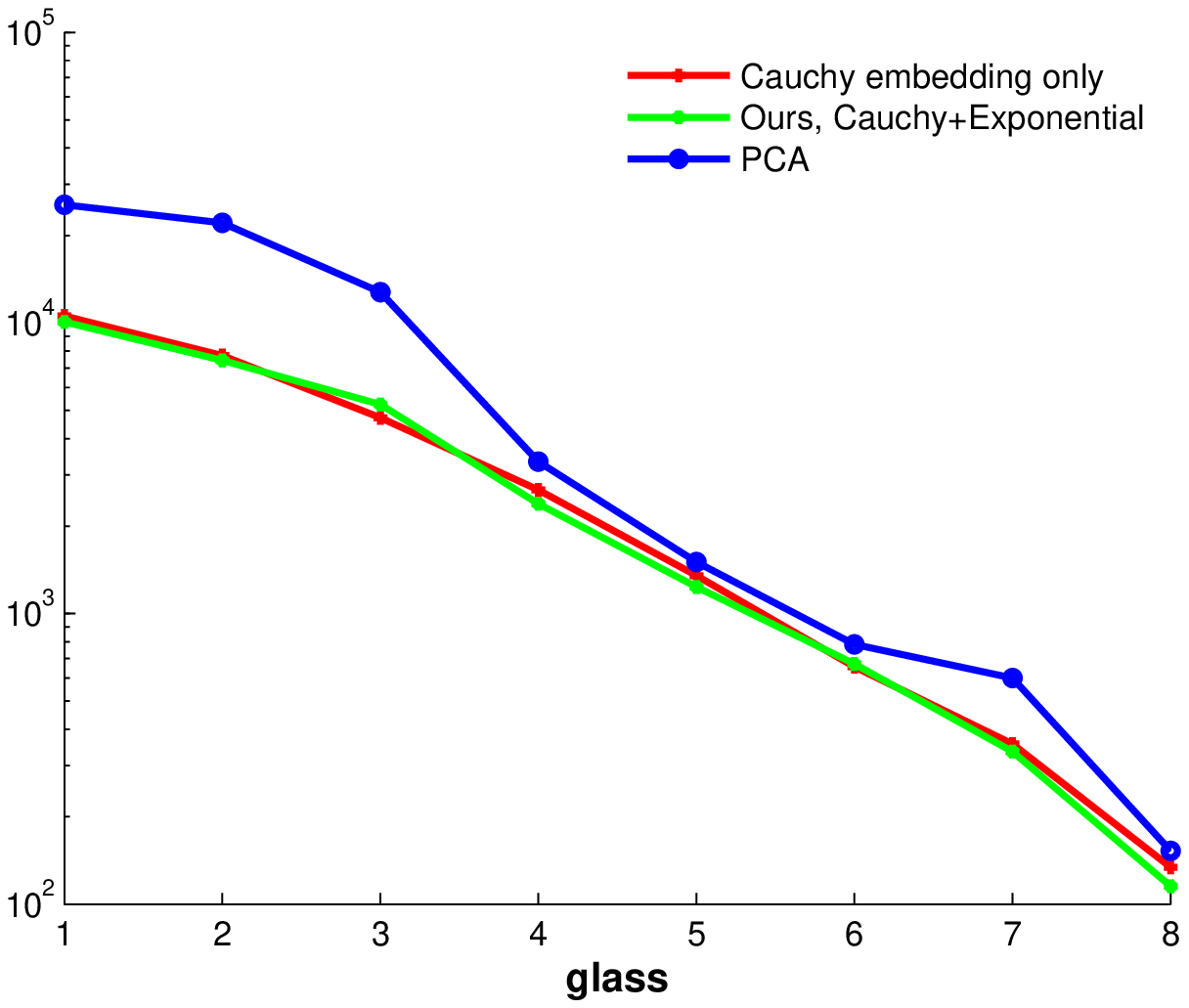}\\
\end{tabular}
\vspace{-0.2in}
\centering
\caption{\small $\ell_1$ norm error v.s. target rank.} \label{fig:ucidata}
\end{figure}

\begin{table}[h!]
\vspace{-0.1in}
\tiny
\caption{ $\ell_1$ rank-$5$ approximation on the synthetic data.}
%\vspace{-0.1in}
\label{tab:syn_data}
\centering
\begin{tabular}{ccccc}
\toprule
& Opt    &   PCA & \cite{swz17} & Ours\\
\midrule
$\ell_1$ loss ($\times 10^4$) & $0.50$  &   $3.53$ &  $1.36$ & $\mathbf{1.04}$\\
\bottomrule
\end{tabular}
\vspace{-0.1in}
\end{table}

%% file: future.tex
\section{Conclusion and Future Work}
We presented an efficient subspace embedding algorithm for orlicz norm and demonstrated its usefulness in regression/low rank approximation problem on synthetic and real datasets. Nevertheless, $O(d \log^{2}n)$ is still a large theoretical approximation factor, and hence it is worth i) investigating whether the theoretical approximation ratio can be smaller if input are under some statistical distribution ii) calculating the actual approximation ratio with ground truth obtained by some slower but more accurate optimization algorithm. It is also worth examining whether our exponential embedding sketching method preserves the statistical properties of the regression error,  since we assumed a different noise distribution from Gaussian/double-exponential as a starting point~\cite{raskutti2014statistical,lopes2018error}.

\section*{Acknowledgements}
Research supported in part by Simons Foundation (\#491119 to Alexandr Andoni), NSF (CCF-1617955, CCF- 1740833), and Google Research Award.

%% file: arxiv_appendix.tex
\appendix

\input{related}

\section{Proof of Fact 2.}
\begin{proof}
Notice that $G_1$ is a nonzero nondecreasing convex function on $\mathbb{R}_+$, thus $G_1^{-1}(1)$ exists, and $G_2$ is a nonzero nondecreasing function. In addition because $s=\sup\left\{\frac{1}{y-x}\left(G_2(y)-G_2(x)\right)\mid 0\leq x\leq y\leq 1\right\},$ $G_2$ is also convex. Thus $\|\cdot\|_{G_2}$ is Orlicz norm. Let $x\in\mathbb{R}^n$. Notice that if $\alpha>0$ satisfies $\sum_{i=1}^n G_1(|x_i|/\alpha)\leq 1,$ then $\forall i\in [n], G_1(|x_i|/\alpha)\leq 1$. It means that $|x_i|\leq G_1^{-1}(1)\alpha$, thus $\sum_{i=1}^n G_2(|x_i|/(G_1^{-1}(1)\alpha))=\sum_{i=1}^n G_1(|x_i|/\alpha)\leq 1$. Similarly if $\alpha$ satisfies $\sum_{i=1}^n G_2(|x_i|/\alpha)\leq 1,$ then $\sum_{i=1}^n G_1(G_1^{-1}(1)|x_i|/\alpha)=\sum_{i=1}^n G_2(|x_i|/\alpha)\leq 1$. Therefore, $\|x\|_{G_1}=\|x\|_{G_2}/G_1^{-1}(1).$
\end{proof}
\section{Proof of Lemma 3.}
Due to convexity of $G$ and $G(1)=1,G(0)=0$, $\forall x\in[0,1],G(x)\leq xG(1)+(1-x)G(0)=x.$ Since $x\leq 1,$ $G(1)/G(x)\leq C_G(1/x)^2,$ we have $G(x)\geq x^2/C_G.$
\section{Proof of Lemma 4.}
With out loss of generality, we can assume $\forall i\in[n], x_i\geq 0.$
Let $x\in\mathbb{R}^n,\alpha=\|x\|_G.$ We have $\sum_{i=1}^n G(x_i/\alpha)= 1.$
If $x_i/\alpha\leq 1,$ due to the convexity of $G$, $G(x_i/\alpha)\leq G(1)\cdot x_i/\alpha+G(0)\cdot(1-x_i/\alpha)=G(1)\cdot x_i/\alpha=x_i/\alpha.$ If $x_i/\alpha>1,$ then $G(x_i/\alpha)>1$ which contradicts to $\sum_{i=1}^n G(x_i/\alpha)= 1.$ Thus, $\|x\|_G\leq \|x\|_1$.

%Let $x_1\geq x_2\geq \cdots \geq x_n$, then
$
\|x/\alpha\|_2^2=\sum_{i=1}^n (x_i/\alpha)^2\leq \sum_{i=1}^n C_G G(x_i/\alpha) =C_G$. Then
\begin{align*}
%\|x/\alpha\|_2^2\leq \sum_{i=1}^n G(x_i/x_1)\leq \sum_{i=1}^n G(x_i/t\alpha)\leq t^{p-1} \sum_{i=1}^n G(x_i/\alpha)\leq t^{p-1}.
\|x\|_2\leq \sqrt{C_G}\alpha.
\end{align*}

\section{Proof of Lemma 5.}
Due to the convexity of $G(\cdot)$ and $G(0)=0,$ $\forall 0<x<y,$ we have $G(x)\leq G(y)x/y+G(0)(1-x/y)=G(y)x/y.$ Thus, $y/x\leq G(y)/G(x).$

\section{Proof of Lemma 6.}
It is easy to see that $\forall x > 0, G(x) \neq 0$, since otherwise for $y>x$, the condition $G(y)/G(x) < C_G (y/x)^2$ would be violated. Let $s=G'_+(1)$. There are several cases.

 If $a \geq 1$ or $ b \geq 1$. Without loss of generality, assume $a \geq 1. G(a)G(b)/G(ab) = (sa - (s-1))G(b)/G(ab) \leq saG(b)/G(ab)$. since $ab \geq b$, $G(ab)/G(b) \geq a$. Therefore, $G(a)G(b)/G(ab)\leq saG(b)/G(ab) \leq s$.

If $a, b \leq 1, 0.5 \leq a \leq 1$ or $0.5 \leq b \leq 1$, we want to show $G(a)G(b)/G(ab)$ is still bounded. Without loss of generality, assume that $0.5 \leq a \leq 1$. Then $G(a)G(b)/G(ab) = G(a)\frac{G(b)}{G(ab)} \leq G(b)/G(ab) \leq C_G/a^2\leq 4C_{G}$.

If $a, b \leq 0.5$ and $G'(0) > 0$.
Let $G'(0) = c > 0$. Therefore, there is a constant $\delta_1$ which may depend on $G$ such that $\forall x \in (0, \delta_1)$, $|\frac{G(x) - G(0)}{x - 0} - c| < \frac{c}{2}$. Therefore, $\forall x \in (0, \delta_1]$, $G(x) > \frac{c}{2}x$. Due to Lemma~5, $\forall x > \delta_1$, $G(x)/x >G(\delta_1)/\delta_1 > c/2$. Therefore, $\forall x, G(x) \geq \frac{c}{2}x.$
Since $b \leq 0.5, ab \leq a \leq 1$. Since $G$ is convex, $G(a)\leq \frac{1-a}{1-ab}G(ab)+\frac{a-ab}{1-ab}G(1)=\frac{1-a}{1-ab}G(ab)+\frac{a-ab}{1-ab}.$ Therefore, $G(a) \leq G(ab)+(1-a)/(1-ab) \leq  G(ab)+2a$.
Similarly, $G(b) \leq 2b + G(ab)$. Then we have $G(a)G(b) \leq (2b + G(ab))(2a + G(ab))$
$\Rightarrow G(a)G(b)/G(ab) \leq ab/G(ab) + (2a + 2b) + G(ab) \leq 2/c + 2 + 1 \leq 2/c + 3$.

If $a, b \leq 0.5, G'_+(0) = 0, G''_+(0) = c_{2} > 0$. Let $\epsilon = c_2/4$. Since $G$ is twice differentiable in $(0,\delta_G)$ and $G'_+(0),G''_+(0)$ exist, by Taylor's Theorem, there is a constant $\delta_2>0$ which may depend on $G$ such that $|G(x) - (G(0) + G'_+(0)x + c_{2}x^{2}/2)| \leq \epsilon x^{2}$. Therefore, $\forall x \in (0, \delta_2), G(x) \geq c_{2}x^{2}/4, G(x) \leq c_{2}x^{2}$. Hence, $\forall a,b \in (0,\delta_2]$, $G(a)G(b)/G(ab) \leq \frac{G(a)G(b)}{c_{2}a^{2}b^{2}/4} \leq \frac{4}{c_{2}}\frac{G(a)}{a^{2}}\frac{G(b)}{b^{2}} \leq \frac{4}{c_{2}}c_{2}^{2} = 4c_{2}$.
Consider $a$ or $b>\delta_2$. Without loss of generality, assume $a > \delta_2$. Similar to the previous argument, $G(a)G(b)/G(ab) \leq G(a)\frac{G(b)}{G(ab)} \leq G(b)/G(ab) \leq C_Gb^{2}/(ab)^{2} \leq C_G/\delta_2^{2}$. Thus $G(a)G(b)/G(ab)$ is bounded by $C_G/\delta_2^2$ in this case.

\section{Proof of Theorem 9}
Without loss of generality, we assume $\forall i\in[n], x_i\geq 0.$ Fix $i\in [n],$ we have
$
\Pr(u_i\geq G^{-1}(1/n^{20}))=e^{-G(G^{-1}(1/n^{20}))}\geq 1-1/n^{20}.
$
%The last inequality follows by Lemma~\ref{lem:func_grow_ratio} and $1/n^{20}\leq 1$.
Define $\mathcal{E}$ to be the event that $\forall i\in[n], u_i\geq G^{-1}(1/n^{20})$. By taking union bound over $n$ coordinates, $\mathcal{E}$ happens with probability at least $1-1/n^{19}.$ Let $\alpha=\|x\|_G$. Then, for any $\gamma\geq 1,$ we have
{\small
\begin{align*}
&\Pr(\|f(x)\|_G\geq \gamma\alpha)\\
=&\Pr(\|f(x)\|_G\geq \gamma\alpha \mid \mathcal{E})\Pr(\mathcal{E})\\
&+\Pr(\|f(x)\|_G\geq \gamma\alpha \mid \neg\mathcal{E})\Pr(\neg\mathcal{E})\\
\leq &\Pr(\|f(x)\|_G\geq \gamma\alpha \mid \mathcal{E})\Pr(\mathcal{E})+\Pr(\neg\mathcal{E})\\
= &\Pr(\|f(x)/(\gamma\alpha)\|_G\geq 1 \mid \mathcal{E})\Pr(\mathcal{E})+\Pr(\neg\mathcal{E})\\
\leq & \E\left(\sum_{i=1}^n G\left(x_i/\alpha\cdot \frac1{\gamma u_i}\right)\mid \mathcal{E}\right)\Pr(\mathcal{E})+1/n^{19}\\
= & \sum_{i=1}^n \E\left(G\left(x_i/\alpha\cdot \frac1{\gamma u_i}\right)\mid \mathcal{E}\right)\Pr(\mathcal{E})+1/n^{19}.
\end{align*}
}
Let $r=G^{-1}(1/n^{20}).$ For a fixed $i\in[n],$
{\small
\begin{align*}
&\E\left(G\left(x_i/\alpha\cdot \frac1{\gamma u_i}\right)\mid \mathcal{E}\right)\Pr(\mathcal{E})\\
=&\int_{r}^{\infty} G\left(\frac{x_i/\alpha}{u\gamma}\right)e^{-G(u)} G'(u) \mathrm{d}u\\
\leq&\frac{1}{\gamma}G(x_i/\alpha)\int_{1}^{\infty} e^{-G(u)} \mathrm{d}G+\frac{1}{\gamma}\int_{r}^{1} G\left(\frac{x_i/\alpha}{u}\right)e^{-G(u)} \mathrm{d}G\\
\leq &\frac{1}{\gamma}G(x_i/\alpha)+\frac{1}{\gamma}\int_{r}^{1} G\left(\frac{x_i/\alpha}{u}\right)e^{-G(u)} \mathrm{d}G\\
\leq &\frac{1}{\gamma}G(x_i/\alpha)+\frac{1}{\gamma}\alpha_GG(x_i/\alpha)\int_{r}^{1} \frac{1}{G(u)}e^{-G(u)} \mathrm{d}G\\
\leq & O(\log n)\frac{\alpha_G}{\gamma}G(x_i/\alpha),
\end{align*}
}
where $\alpha_G$ is a constant may depend on $G(\cdot)$.
The first inequality follows by $G(x_i/\alpha\cdot 1/(\gamma u))\leq 1/\gamma \cdot G(x_i/\alpha\cdot 1/u)+(1-1/\gamma)\cdot G(0)\leq G(x_i/\alpha\cdot 1/u)/\gamma\leq G(x_i/\alpha)/\gamma$. The second inequality follows by $\int_{r}^{\infty} e^{-x}\mathrm{d}x\leq 1$. The third inequality follows by Lemma 6. Since $x_i/\alpha\leq 1$, then there is an $\alpha_G$ such that $G(u)G(x_i/\alpha\cdot 1/u)\leq \alpha_G G(x_i/\alpha).$ The last inequality follows by $\int_{1/n^{20}}^\infty e^{-x}/x \mathrm{d}x=O(\log n).$

Thus, we have
{\small
\begin{align*}
&\sum_{i=1}^n \E\left(G\left(x_i/\alpha\cdot \frac1{\gamma u_i}\right)\mid \mathcal{E}\right)\Pr(\mathcal{E})\\
\leq&  O(\log n)\frac{\alpha_G}{\gamma}\sum_{i=1}^n G(x_i/\alpha)\leq O(\log n)\frac{\alpha_G}{\gamma}.
\end{align*}
}
Then,
\begin{align*}
\Pr(\|f(x)\|_G\geq \gamma\alpha)\leq O(\log n)\frac{\alpha_G}{\gamma}+1/n^{19}.
\end{align*}
It is equivalent to
\begin{align*}
\Pr(\|f(x)\|_G\leq \gamma\alpha)\geq 1-O(\log n)\frac{\alpha_G}{\gamma}-1/n^{19}.
\end{align*}
Set $\gamma=O(\log n)\frac{\alpha_G}{\delta}$, we complete the proof.

\section{Proof of Theorem 10}
Similar to~\cite{ddhkm09}, we can define a well conditioned basis for Orlicz norm.
\begin{definition}[Well conditioned basis for Orlicz norm]\label{def:well_condition}
Given a matrix $A\in\mathbb{R}^{n\times m}$ with rank $d$, let $U\in\mathbb{R}^{n\times d}$ be a matrix which has the same column space of $A$. If $U$ satisfies
1. $\forall x\in\mathbb{R}^d,$ $\|x\|_\infty\leq \beta \|Ux\|_G$,
2. $\sum_{i=1}^d \|U_i\|_G\leq \alpha$ ,
then $U$ is an $(\alpha,\beta,G)$-well conditioned basis of $A$.
\end{definition}
Fortunately, the such good basis exists for Orlicz norm.
\begin{theorem}[See Connection to Auerbach basis in Section 3.1 of~\cite{ddhkm09}]
Given a matrix $A\in\mathbb{R}^{n\times m}$ with rank $d$ and norm $\|\cdot\|_G$, there exist a matrix $U\in\mathbb{R}^{n\times d}$ which is a $(d,1,G)$ well conditioned basis of $A$.
\end{theorem}

\begin{proof}[Proof of Theorem 10]
Notice that $D^{-1}Ax$ is exactly the same as $f(Ax)$. %Due to Lemma~\ref{def:well_condition},
There is a matrix $U\in\mathbb{R}^{n\times d}$ which is $(d,1,G)$-well conditioned basis of $A$. Since $\forall x\in \mathbb{R}^m,$ there is always a vector $y\in\mathbb{R}^d$ such that $Ax=Uy,$ we only need to prove that with probability at least $0.99,$
\begin{align*}
\forall x\in\mathbb{R}^d, \|D^{-1}Ux\|_G\leq O(\alpha_Gd^2\log n)\|Ux\|_G,
\end{align*}
where $D,\alpha_G$ are the same as stated in the Theorem. According to Theorem~9, if we look at a fixed $i\in[d],$ then with probability at least $1-0.01/d,$ $\|D^{-1}U_i\|_G\leq O(\alpha_Gd\log(n)).$ By taking union bound, with probability at least $0.99$, $\forall i\in[d],$ $\|D^{-1}U_i\|_G\leq O(\alpha_Gd\log(n)).$ Now we have, for any $x\in\mathbb{R}^d,$
\begin{align*}
\|D^{-1}Ux\|_G&\leq \sum_{i=1}^d |x_i| \|D^{-1}U_i\|_G\\
&\leq \|x\|_{\infty} \sum_{i=1}^d \|D^{-1}U_i\|_G\\
&\leq O(\alpha_Gd\log(n))\|x\|_{\infty} \sum_{i=1}^d \|U_i\|_G\\
&\leq O(\alpha_Gd^2\log(n))\|Ux\|_{G}.
\end{align*}
The first inequality follows by triangle inequality. The third inequality follows by $\forall i\in[d],$ $\|D^{-1}U_i\|_G\leq O(\alpha_Gd\log(n)).$ The forth inequality follows by $(d,1,G)$-well conditioned basis.
\end{proof}

\section{Proof of Theorem 12}
Now, in the following, we present the concept of $\varepsilon$-net.
\begin{definition}[$\varepsilon$-net for $\ell_2$ norm]
Given $A\in\mathbb{R}^{n\times m}$ with rank $d$, let $S$ be the $\ell_2$ unit ball in the column space of $A$, i.e. $S=\{y\mid \|y\|_2=1,\exists x\in\mathbb{R}^m,y=Ax\}.$ Let $N\subset S$, if $\forall x\in S,\exists y\in N$ such that $\|x-y\|_2\leq \varepsilon,$ then we say $N$ is an $\varepsilon$-net for $S$.
\end{definition}
The following theorem gives an upper bound of the size of $\varepsilon$-net.
\begin{theorem}[Lemma 2.2 of~\cite{w14}]\label{thm:eps_net}
Given $A\in\mathbb{R}^{n\times m}$ with rank $d$, let $S$ be the $\ell_2$ unit ball in the column space of $A$. There exist an $\varepsilon$-net $N$ for $S$, such that $|N|\leq (1+4/\varepsilon)^d.$
\end{theorem}

It suffices to prove $\forall x\in\mathbb{R}^m,\|Ax\|_2=1$ we have $ \Omega(1/(\alpha'_G d\log n))\|Ax\|_G\leq \|D^{-1}Ax\|_{\infty}$.
Let $D\in\mathbb{R}^{n\times n}$ be a diagonal matrix of which each entry on the diagonal is an i.i.d. random variable drawn from the distribution with CDF $1-e^{-G(t)}.$ Let $\alpha'_G\geq 1$ be a sufficiently large constant. Let $S$ be the $\ell_2$ unit ball in the column space of $A$. Let $t_1=\Theta(\alpha'_G d\log n),t_2=\Theta(\alpha_G d^2\log n),$ where $\alpha_G$ is the parameter stated in Theorem 10. Set $\varepsilon=O(1/(\sqrt{n}C_Gt_1t_2))$. There exist an $\varepsilon$-net $N$ for $S$, and
\begin{align*}
|N|=e^{O(d(\log n+\log(C_G\alpha'_G\alpha_G)))}.
\end{align*}
By taking union bound over the net points, according to Theorem 11, with probability at least $0.99$,
\begin{align}\label{eq:no_contraction_for_net}
\forall x\in N, \|D^{-1}x\|_\infty\geq \Omega(1/(\alpha'_G d\log n))\|x\|_G.
\end{align}
Also due to Theorem 10, with probability at least $0.99$,
\begin{align}\label{eq:no_dilation_for_all}
\forall x\in S, \|D^{-1}x\|_G\leq O(\alpha_Gd^2\log n)\|x\|_G.
\end{align}
By taking union bound, with probability at least $0.98$, the above two events will happen. Then, in this case, consider a $y\in S,$ let $x\in N$ such that $\|x-y\|_2\leq \varepsilon,$ let $z=x-y,$ we have
\begin{align*}
\|D^{-1}y\|_\infty&\geq \|D^{-1}x\|_\infty-\|D^{-1}z\|_\infty\\
&\geq \frac{1}{t_1} \|x\|_G-t_2\sqrt{C_G}\|z\|_G\\
&\geq \frac{1}{t_1} \|y\|_G-\frac{t_2}{t_1}\|z\|_G-t_2\sqrt{C_G}\|z\|_G\\
&\geq \frac{1}{t_1} \|y\|_G-2t_2\sqrt{C_G}\|z\|_G\\
&\geq \frac{1}{t_1} \|y\|_G-O(\frac{2}{\sqrt{C_G}t_1})\\
&\geq \Omega(1/t_1) \|y\|_G\\
&=\Omega(1/(\alpha'_G d\log n)) \|y\|_G.
\end{align*}
The first inequality follows by triangle inequality. The second inequality follows by Equation~\ref{eq:no_contraction_for_net}, Equation~\ref{eq:no_dilation_for_all}, and Lemma 4, i.e. $\|D^{-1}z\|_\infty\leq \|D^{-1}z\|_2\leq \sqrt{C_G}\|D^{-1}z\|_G$. The third inequality follows by triangle inequality. The forth inequality follows by $t_1,C_G\geq 1$. The fifth inequality follows by Lemma 4: $\|z\|_G\leq \|z\|_1\leq \sqrt{n}\|z\|_2=\sqrt{n}\varepsilon=O(1/(C_Gt_1)).$ The sixth inequality follows by Lemma 4: $\|y\|_G\geq \frac{1}{\sqrt{C_G}}\|y\|_2=1/\sqrt{C_G}.$

\section{Proof of Theorem 13}
Due to Theorem~12 and Theorem~10, with probability at least $0.98,$ $\forall x\in\mathbb{R}^m,$
$\Omega(1/(\alpha'_G d\log n))\|Ax\|_G\leq \|D^{-1}Ax\|_\infty\leq \|D^{-1}Ax\|_2.$
And
$\|D^{-1}Ax\|_2\leq \sqrt{C_G}\|D^{-1}Ax\|_G\leq O(\sqrt{C_G}\alpha'_Gd^2\log n)\|Ax\|_G.$

\section{Proof of Theorem 16}
Due to Theorem 14 and Theorem 15, with probability at least $0.95$, $\forall x,$ $\|\Pi_2\Pi_1D^{-1}Ax\|_2$ is a constant approximation to $\|\Pi_1D^{-1}Ax\|_2$ and $\|\Pi_1D^{-1}Ax\|_2$ is a constant approximation to $\|D^{-1}Ax\|_2.$ Combining with Theorem 13, we complete the proof.

\section{Proof of Theorem 18}
Let $x^*=\arg\min_{x\in\mathbb{R}^d}\|Ax-b\|_G.$ Due to Theorem~9, with probability at least $0.99,$
\begin{align}\label{eq:for_each_no_dilation}
\|D^{-1}(Ax^*-b)\|_G\leq O(\alpha_G \log n) \|Ax^*-b\|_G.
 \end{align}
 Now let $A'=[A\ b]\in\mathbb{R}^{n\times (d+1)}.$ Due to Theorem~16, with probability at least $0.9,$ we have
\begin{align}\label{eq:for_all_no_contraction}
\forall x\in\mathbb{R}^{d+1}, \Omega(1/(\alpha'_G d\log n))\|A'x\|_G\leq \|\Pi_2\Pi_1D^{-1}A'x\|_2.
\end{align}
Then,
\begin{align*}
\|A\hat{x}-b\|_G&\leq O(\alpha'_G d\log n)\|\Pi_2\Pi_1D^{-1}(A\hat{x}-b)\|_2\\
&\leq O(\alpha'_G d\log n)\|\Pi_2\Pi_1D^{-1}(Ax^*-b)\|_2\\
&\leq O(\alpha'_G d\log n)\|D^{-1}(Ax^*-b)\|_2\\
&\leq O(\alpha'_G\sqrt{C_G} d\log n)\|D^{-1}(Ax^*-b)\|_G\\
&\leq O(\alpha_G\alpha'_G\sqrt{C_G} d\log^2 n)\|Ax^*-b\|_G.
\end{align*}
The first inequality follows by Equation~\ref{eq:for_all_no_contraction}. The second inequality follows by $\hat{x}=(\Pi_2\Pi_1D^{-1}A)^\dagger \Pi_2\Pi_1D^{-1}b,$ which is the optimal solution for $\min_{x\in\mathbb{R}^d}\|\Pi_2\Pi_1D^{-1}(Ax-b)\|_2.$ The third inequality follows by Theorem~14 and Theorem~15. The forth inequality follows by Lemma~4. The last inequality follows by Equation~\ref{eq:for_each_no_dilation}.
Let $\beta_G=\alpha'_G\sqrt{C_G},$ we complete the proof of the correctness of Algorithm~1.

For the running time, according to Theorem~16, computing $\Pi_2\Pi_1D^{-1}A$ and $\Pi_2\Pi_1D^{-1}b$ needs $\mathrm{nnz(A)}+\mathrm{poly}(d)$ time. Since $\Pi_2\Pi_1D^{-1}A$ has size $\mathrm{poly}(d),$ computing $\hat{x}=(\Pi_2\Pi_1D^{-1}A)^\dagger \Pi_2\Pi_1D^{-1}b$ needs $\mathrm{poly}(d)$ running time. The total running time is $\mathrm{nnz(A)}+\mathrm{poly}(d).$

\section{proof of Lemma 20}
Before we prove the Lemma, we need following tools.

\begin{lemma}[Concentration bound for sum of half normal random variables]\label{lem:half_normal_concentration}
For any $k$ i.i.d. random Gaussian variables $z_1,z_2,\cdots,z_k,$ we have that
{\small
\begin{align*}
\Pr\left(\frac{1}{k}\sum_{i=1}^k |z_i|\in\left((1-\varepsilon)\sqrt{\frac{2}{\pi}},(1+\varepsilon)\sqrt{\frac{2}{\pi}}\right)\right)\geq 1-e^{-\Omega(k\varepsilon^2)}.
\end{align*}
}
\end{lemma}

\begin{lemma}\label{lem:upperbound_spectrum}
Let $G\in\mathbb{R}^{k\times m}$ be a random matrix with each entry drawn uniformly from i.i.d. $N(0,1)$ Gaussian distribution. With probability at least $0.99,$ $|G|_2\leq 10\sqrt{km}.$
\end{lemma}
\begin{proof}
Since $\E\left(\|G\|_F^2=km\right),$ we have that $\Pr(\|G\|_F^2\geq 100km)\leq 0.01$. Thus, with probability at least $0.99,$ we have $\|G\|_2\leq \|G\|_F\leq 10\sqrt{km}.$
\end{proof}

Now, let us prove the lemma.
\begin{proof}[Proof of Lemma 20]
Without loss of generality, we only need to prove $\forall x\in\mathbb{R}^n$ with $\|x\|_2=1,$ we have $\|Bx\|_1\in(1-\varepsilon,1+\varepsilon).$  Let set $S=\{v\mid v\in\mathbb{R}^n,\|v\|_2=1\}.$ Due to Theorem~\ref{thm:eps_net}, we can find a set $G\subset S$ which satisfies that $\forall u\in S$ there exists $v\in G$ such that $\|u-v\|_2\leq (\varepsilon/(1000n))^{10}$ and $|G|\leq (4000n/\varepsilon)^{20n}.$ Let $k\geq c\varepsilon^{-2}n\ln(n/\varepsilon)$ where $c$ is a sufficiently large constant. By Lemma~\ref{lem:half_normal_concentration}, we have that for a fixed $v\in G,$ with probability at least $1-e^{-1000n\ln(4000n/\varepsilon)},\|Bv\|_1\in(1-\varepsilon,1+\varepsilon).$ By taking union bound over all the points in $G$, we have
$
\Pr\left(\forall v\in G,\|Bv\|_1\in(1-\varepsilon,1+\varepsilon)\right)\geq 1-e^{-980n\ln(4000n/\varepsilon)}\geq 0.99.
$
Now, consider $\forall x\in \mathbb{R}^n$ with $\|x\|_2=1,$ i.e. $x\in S,$ we can find $v\in G$ such that $\|v-x\|_2\leq (\varepsilon/(1000n))^{10},$ and let $u=v-x.$ Then, conditioned on 
\begin{align*}
\|B\|_2\leq 10\sqrt{tn}\cdot\sqrt{\pi/2}/t,
\end{align*}
we have
{
\begin{align*}
&\|Bx\|_1\\
\in&(\|Bv\|_1-\|Bu\|_1,\|Bv\|_1+\|Bu\|_1)\\
\subseteq&(1-(\varepsilon+\sqrt{t}\|B\|_2\|u\|_2),1+(\varepsilon+\sqrt{t}\|B\|_2\|u\|_2))\\
\subseteq&(1-2\varepsilon,1+2\varepsilon)
\end{align*}
}
where the first relation follows by triangle inequality, the second relation follows by 
\begin{align*}
\|Bu\|_1\leq \sqrt{t}\|Bu\|_2\leq \sqrt{t}\|B\|_2\|u\|_2,
\end{align*}
and the last relation follows by 
\begin{align*}
\|u\|_2\leq (\varepsilon/(1000n))^{10},\|B\|_2\leq 10\sqrt{tn}\cdot \sqrt{\pi/2}/t.
\end{align*}

According to Lemma~\ref{lem:upperbound_spectrum}, we know that with probability at least $0.99,$ we have 
\begin{align*}\|B\|_2\leq 10\sqrt{tn}\cdot\sqrt{\pi/2}/t.
 \end{align*}
 By taking union bound, we have with probability at least $0.98,$ 
 \begin{align*}
 \forall x\in S, \|Bx\|_1\in (1-2\varepsilon,1+2\varepsilon).
  \end{align*}
  By adjusting the $\varepsilon,$ we complete the proof.
\end{proof}

\section{Proof of Theorem 21}
Without loss of generality, we assume constant $k\leq 2.$ Otherwise, we can always adjust constants in all the related theorems and lemmas to make larger $k$ work.
Let $x^*=\arg\min_{x\in\mathbb{R}^d}\sum_{i=1}^k\|A_ix-b_i\|_{G_i}.$ By Theorem~9 and taking union bound, we have that with probability at least $0.98,$
{\small
\begin{align}
&\forall i\in\{1,2,\cdots,k\},\notag\\
&\|(D^{(i)})^{-1}(A_ix^*-b_i)\|_{G_i}\leq O(\alpha_{G_i} \log n) \|A_ix^*-b_i\|_{G_i}.\label{eq:for_each_no_dilation_2}
\end{align}
}
 Now let $A'_i=[A_i\ b_i]\in\mathbb{R}^{n_i\times (d+1)}.$ Due to Theorem~16 and union bound, with probability at least $0.8,$ we have
 {
 \small
\begin{align}
&\forall x\in\mathbb{R}^{d+1},\notag \\
 &\Omega(1/(\alpha'_{G_i} d\log n_i))\|A'_ix\|_{G_i}
\leq \|\Pi_2^{(i)}\Pi_1^{(i)}(D^{(i)})^{-1}A'_ix\|_2.\label{eq:for_all_no_contraction_2}
\end{align}
}
Then,
{
\small
\begin{align*}
&\sum_{i=1}^k\|A_i\hat{x}-b_i\|_{G_i}\\
\leq& \sum_{i=1}^k O(\alpha'_{G_i} d\log n)\|\Pi_2^{(i)}\Pi_1^{(i)}(D^{(i)})^{-1}(A_i\hat{x}-b_i)\|_2\\
\leq& \sum_{i=1}^k O(\alpha'_{G_i} d\log n)\|B^{(i)}\Pi_2^{(i)}\Pi_1^{(i)}(D^{(i)})^{-1}(A_i\hat{x}-b_i)\|_1\\
\leq& O(\max_{i\in[k]}\alpha'_{G_i} d\log n) \|B\Pi_2\Pi_1D^{-1}(A\hat{x}-b)\|_1\\
\leq& O(\max_{i\in[k]}\alpha'_{G_i} d\log n) \|B\Pi_2\Pi_1D^{-1}(Ax^*-b)\|_1\\
\leq& O(\max_{i\in[k]}\alpha'_{G_i} d\log n)\sum_{i=1}^k \|B^{(i)}\Pi_2^{(i)}\Pi_1^{(i)}(D^{(i)})^{-1}(A_ix^*-b_i)\|_1\\
\leq& O(\max_{i\in[k]}\alpha'_{G_i} d\log n)\sum_{i=1}^k \|\Pi_2^{(i)}\Pi_1^{(i)}(D^{(i)})^{-1}(A_ix^*-b_i)\|_2\\
\leq& O(\max_{i\in[k]}\alpha'_{G_i} d\log n)\sum_{i=1}^k \|(D^{(i)})^{-1}(A_ix^*-b_i)\|_2\\
\leq& O((\max_{i\in[k]}\sqrt{C_{G_i}})(\max_{i\in[k]}\alpha'_{G_i}) d\log n)\sum_{i=1}^k\|(D^{(i)})^{-1}(A_ix^*-b_i)\|_{G_i}\\
\leq& O((\max_{i\in[k]}\alpha_{G_i})(\max_{i\in[k]}\sqrt{C_{G_i}})(\max_{i\in[k]}\alpha'_{G_i}) d\log^2 n)\sum_{i=1}^k\|A_ix^*-b_i\|_{G_i}.
\end{align*}
}
The first inequality follows by Equation~\ref{eq:for_all_no_contraction_2}. The second inequality follows by Lemma~20. The forth inequality follows by $\hat{x}$ is the optimal solution for $\min_{x\in\mathbb{R}^d}\|B\Pi_2\Pi_1D^{-1}(Ax-b)\|_1.$ The sixth inequality follows by Lemma~20. The seventh inequality follows by Theorem~14 and Theorem~15. The eighth inequality follows by Lemma~4. The last inequality follows by Equation~\ref{eq:for_each_no_dilation_2}.
Let $\beta'_G=(\max_{i\in[k]}\alpha_{G_i})(\max_{i\in[k]}\sqrt{C_{G_i}})(\max_{i\in[k]}\alpha'_{G_i}),$ we complete the proof of the correctness of Algorithm~2.

For the running time, according to Theorem~16, computing $\Pi_2\Pi_1D^{-1}A$ and $\Pi_2\Pi_1D^{-1}b$ needs $\sum_{i=1}^k\mathrm{nnz(A_i)}+\mathrm{poly}(d)$ time. Due to Lemma 20, the size of $B$ is $\mathrm{poly}(d).$ To compute $B\Pi_2\Pi_1D^{-1}A$ and $B\Pi_2\Pi_1D^{-1}b,$ we need additional $\mathrm{poly}(d)$ time. Since $B\Pi_2\Pi_1D^{-1}A$ has size $\mathrm{poly}(d),$ computing the optimal solution of $\min_{x\in\mathbb{R}^d}\|B\Pi_2\Pi_1D^{-1}(Ax-b)\|_1$ by using linear programming needs $\mathrm{poly}(d)$ running time. The total running time is $\sum_{i=1}^k\mathrm{nnz(A_i)}+\mathrm{poly}(d).$

\section{Proof of Theorem 23}
Before we prove the Theorem, we need to show following Lemmas.
\begin{lemma}[\cite{swz17}]\label{lem:zero_step}
Let $A\in\mathbb{R}^{n\times d},R\in\mathbb{R}^{d\times t_3},k$ be the same as in the Algorithm~3, then with probability at least $0.9,$
{\small
\begin{align*}
&\min_{X\in\mathbb{R}^{t_3\times k},Y\in\mathbb{R}^{k\times d}}\|ARXY-A\|_p^p\\
\leq& O((k\log k)^{1-p/2}\log n)\min_{U\in\mathbb{R}^{n\times k},V\in\mathbb{R}^{k\times d}} \|UV-A\|_p^p.
\end{align*}
}
\end{lemma}
\begin{lemma}[\cite{wz13}]\label{lem:stronger_gaurantee}
Let $1\leq p\leq 2.$ Given a matrix $A\in\mathbb{R}^{n\times d}$, $d\leq n$, let $D\in\mathbb{R}^{n\times n}$ be a diagonal matrix of which each entry on the diagonal is an i.i.d. random variable drawn from the distribution with CDF $1-e^{-t^p}.$ Let $\Pi_1\in\mathbb{R}^{t_1\times n}$ be a sparse embedding matrix (see Theorem~18) and let $\Pi_2\in\mathbb{R}^{t_2\times t_1}$ be a random Gaussian matrix (see Theorem~19) where $t_1=\Omega(d^2),t_2=\Omega(d).$ Then, with probability at least $0.9,$
{\small
\begin{align*}
&\forall x\in\mathbb{R}^d,\\
&\Omega(1/\min\{(d\log d)^{1/p},(d\log d\log n)^{1/p-1/2}\})\|Ax\|_p\\
\leq &\|\Pi_2\Pi_1D^{-1}Ax\|_2.
\end{align*}
}
\end{lemma}
\begin{lemma}\label{lem:first_step}
Let $A\in\mathbb{R}^{n\times d},S\in\mathbb{R}^{t_2\times n},R\in\mathbb{R}^{d\times t_3},k$ be the same as in the Algorithm~3, then with probability at least $0.9,$
{\small
\begin{align*}
&\min_{X\in\mathbb{R}^{t_2\times k},Y\in\mathbb{R}^{k\times t_2}}\|ARXYSA-A\|_p^p\\
\leq& \beta\min_{U\in\mathbb{R}^{n\times k},V\in\mathbb{R}^{k\times d}} \|UV-A\|_p^p,
\end{align*}
}
where
{\small
\begin{align*} 
\beta=O(\min((k\log k)^{2-p/2}\log^{p+1} n,(k\log k)^{2-p}\log^{2+p/2}n)).
\end{align*}
}
\end{lemma}
\begin{proof}
Let 
{\small
\begin{align*}
X^*,V^*=\arg\min_{X\in\mathbb{R}^{t_3\times k},V\in\mathbb{R}^{k\times d}}\|ARXV-A\|_p^p.
 \end{align*}
}
 Let $U^*=ARX^*,\tilde{V}=(SU^*)^\dagger SA$. Let $\gamma=\min\{k\log k,(k\log k\log n)^{1-p/2}\}.$ We have
 {\small
\begin{align*}
&\|U^*\tilde{V}-A\|_p^p\\
\leq& 2^{p-1}\|U^*(\tilde{V}-V^*)\|_p^p+2^{p-1}\|U^*V^*-A\|_p^p\\
\leq& O(\gamma)\sum_{i=1}^d \|SU^*(\tilde{V}-V^*)_i\|_2^p+2^{p-1}\|U^*V^*-A\|_p^p\\
\leq& O(\gamma)\sum_{i=1}^d (\|S(U^*\tilde{V}-A)_i\|_2+\|S(U^*V^*-A)_i\|_2)^p\\
&+2^{p-1}\|U^*V^*-A\|_p^p\\
\leq& O(\gamma)\sum_{i=1}^d (2\|S(U^*V^*-A)_i\|_2)^p+2^{p-1}\|U^*V^*-A\|_p^p\\
\leq& O(\gamma)\sum_{i=1}^d\|D_1^{-1}(U^*V^*-A)_i\|_2^p+2^{p-1}\|U^*V^*-A\|_p^p\\
\leq& O(\gamma)\|D_1^{-1}(U^*V^*-A)\|_p^p+2^{p-1}\|U^*V^*-A\|_p^p\\
\leq& O(\gamma)\log^p(nd)\|U^*V^*-A\|_p^p+2^{p-1}\|U^*V^*-A\|_p^p\\
=&O(\gamma\log^p(n))\|U^*V^*-A\|_p^p.
\end{align*}
}
The first inequality follows by convexity of $x^p.$ The second inequality follows by Lemma~\ref{lem:stronger_gaurantee}. The third inequality follows by triangle inequality. The forth inequality follows by $\tilde{V}=(SU^*)^\dagger SA.$ The fifth inequality follows by Theorem~14 and Theorem~15. The sixth inequality follows by $p\leq 2.$ The seventh inequality follows by Theorem~9.

Due to Lemma~\ref{lem:zero_step}, we have
{\small
\begin{align*}
&\|U^*V^*-A\|_p^p\\
\leq& O((k\log k)^{1-p/2}\log n)\min_{U\in\mathbb{R}^{n\times k},V\in\mathbb{R}^{k\times d}}\|UV-A\|_p^p.
\end{align*}
}
Thus, we have
{\small
\begin{align*}
&\min_{X,Y}\|ARXYSA-A\|_p^p\leq \|U^*\tilde{V}-A\|_p^p \\
\leq& O(\min((k\log k)^{2-p/2}\log^{p+1} n,(k\log k)^{2-p}\log^{2+p/2}n))\\
&\cdot\min_{U,V} \|UV-A\|_p^p.
\end{align*}
}
\end{proof}

\begin{lemma}[\cite{swz17}]\label{lem:middle_step}
Let $A\in\mathbb{R}^{n\times d},S\in\mathbb{R}^{t_2\times n},R\in\mathbb{R}^{d\times t_3},k,T_2\in\mathbb{R}^{d\times t_3}$ be the same as in the Algorithm~3, then with probability at least $0.9,$ if for $\alpha\geq 1,$ $\tilde{X},\tilde{Y}$ satisfy
\begin{align*}
\|AR\tilde{X}\tilde{Y}SAT_2-AT_2\|_p^p\leq \alpha \min_{X,Y}\|ARXYSAT_2-AT_2\|_p^p,
\end{align*}
then
\begin{align*}
\|AR\tilde{X}\tilde{Y}SA-A\|_p^p\leq \alpha O(\log n)\min_{X,Y}\|ARXYSA-A\|_p^p.
\end{align*}
\end{lemma}

\begin{lemma}\label{lem:final_step}
Let $A\in\mathbb{R}^{n\times d},S\in\mathbb{R}^{t_2\times n},R\in\mathbb{R}^{d\times t_2},k,T_1\in\mathbb{R}^{t_2\times n},T_2\in\mathbb{R}^{d\times t_3}$ be the same as in the Algorithm~3, then with probability at least $0.9,$ if for $\alpha\geq 1$
\begin{align*}
&\sum_{i=1}^{t_3}\|T_1(AR\tilde{X}\tilde{Y}SAT_2-AT_2)_i\|_2^p\\
&\leq \alpha \min_{X,Y}\sum_{i=1}^{t_3}\|T_1(ARXYSAT_2-AT_2)_i\|_2^p,
\end{align*}
then
\begin{align*}
\|AR\tilde{X}\tilde{Y}SAT_2-AT_2\|_p^p\leq \alpha \beta\min_{X,Y}\|ARXYSAT_2-AT_2\|_p^p,
\end{align*}
where $\beta=O(\min(k\log k\log^p n,(k\log k)^{1-p/2}\log^{1+p/2} n)).$
\end{lemma}
\begin{proof}
Let 
\begin{align*}
X^*,Y^*=\arg\min_{X,Y}\sum_{i=1}^{t_3}\|T_1(ARXYSAT_2-AT_2)_i\|_2^p.
 \end{align*}
 Let $L=AR,N=SAT_2,M=AT_2.$ Let $\gamma=\min\{k\log k,(k\log k\log n)^{1-p/2}\}.$ Let $\tilde{H}=\tilde{X}\tilde{Y}$ and let $H^*=X^*Y^*.$
We have
{\small
\begin{align*}
&\|L\tilde{H}N-M\|_p^p\\
\leq & 2^{p-1} \|L\tilde{H}N-LH^*N\|_p^p+2^{p-1}\|LH^*N-M\|_p^p\\
\leq & O(\gamma) \sum_{i=1}^{t_3}\|T_1(L\tilde{H}N-LH^*N)_i\|_2^p+2^{p-1}\|LH^*N-M\|_p^p\\
\leq & O(\gamma) \sum_{i=1}^{t_3}(\|T_1(L\tilde{H}N-M)_i\|_2+\|T_1(LH^*N-M)_i\|_2)^p\\
&+2^{p-1}\|LH^*N-M\|_p^p\\
\leq & O(\gamma) \sum_{i=1}^{t_3}(\|T_1(L\tilde{H}N-M)_i\|_2^p+\|D_2^{-1}(LH^*N-M)_i\|_2^p)\\
&+2^{p-1}\|LH^*N-M\|_p^p\\
\leq & O(\gamma) (\sum_{i=1}^{t_3}\|T_1(L\tilde{H}N-M)_i\|_2^p+\|D_2^{-1}(LH^*N-M)\|_p^p)\\
&+2^{p-1}\|LH^*N-M\|_p^p\\
\leq & O(\gamma) (\alpha\sum_{i=1}^{t_3}\|T_1(LH^*N-M)_i\|_2^p+\|D_2^{-1}(LH^*N-M)\|_p^p)\\
&+2^{p-1}\|LH^*N-M\|_p^p\\
\leq & O(\gamma) (\alpha\sum_{i=1}^{t_3}\|D_2^{-1}(LH^*N-M)_i\|_2^p+\|D_3^{-1}(LH^*N-M)\|_p^p)\\
&+2^{p-1}\|LH^*N-M\|_p^p\\
\leq & O(\gamma) \alpha\|D_2^{-1}(LH^*N-M)\|_p^p+2^{p-1}\|LH^*N-M\|_p^p\\
\leq &O(\gamma\log^p(n)) \alpha\|LH^*N-M\|_p^p.
\end{align*}
}
The first inequality follows by convexity of $x^p$. The second inequality follows by Lemma~\ref{lem:stronger_gaurantee}. The third inequality follows by triangle inequality. The forth inequality follows by convexity of $x^p$, Theorem~14 and Theorem~15. The fifth inequality follows by $p\leq 2.$ The sixth inequality follows by the property of $\tilde{X},\tilde{Y}.$ The seventh inequality follows by Theorem~14 and Theorem~15. The eighth inequality follows by $p\leq 2.$ Then the ninth inequality follows by Theorem~9.
\end{proof}

Now let us prove Theorem:
\begin{proof}
Notice that 
{\small
\begin{align*}
\hat{X},\hat{Y}=\arg\min_{X\in\mathbb{R}^{t_2\times k},Y\in\mathbb{R}^{k\times t_3}}\|T_1ARXYSAT_2-T_1AT_2\|_F^2,
\end{align*}
 }
 we have
 {\small
\begin{align*}
&(\sum_{i=1}^{t_3}\|T_1(AR\hat{X}\hat{Y}SAT_2-AT_2)_i\|_2^p)^{1/p}\\
\leq&  O((k\log k)^{1/p-1/2})(\min_{X,Y}\sum_{i=1}^{t_3}\|T_1(ARXYSAT_2-AT_2)_i\|_2^p)^{\frac1p}.
\end{align*}
}
It means
{\small
\begin{align*}
&(\sum_{i=1}^{t_3}\|T_1(AR\hat{X}\hat{Y}SAT_2-AT_2)_i\|_2^p)\\
\leq&  O((k\log k)^{1-p/2})(\min_{X,Y}\sum_{i=1}^{t_3}\|T_1(ARXYSAT_2-AT_2)_i\|_2^p).
\end{align*}
}
According to Lemma~\ref{lem:final_step}, we have
\begin{align*}
\|AR\hat{X}\hat{Y}SAT_2-AT_2\|_p^p\leq \beta_1\min_{X,Y}\|ARXYSAT_2-AT_2\|_p^p,
\end{align*}
where 
\begin{align*}\beta_1=O(\min((k\log k)^{2-p/2}\log^p n,(k\log k)^{2-p}\log^{1+p/2} n)).
\end{align*}
Due to Lemma~\ref{lem:middle_step}, we have
\begin{align*}
\|AR\hat{X}\hat{Y}SA-A\|_p^p\leq O(\beta_1\log n)\min_{X,Y}\|ARXYSA-A\|_p^p.
\end{align*}
Then, according to Lemma~\ref{lem:first_step}, we have
\begin{align*}
\|AR\hat{X}\hat{Y}SA-A\|_p^p\leq \beta_2\min_{U\in\mathbb{R}^{n\times k},V\in\mathbb{R}^{k\times d}}\|UV-A\|_p^p,
\end{align*}
where 
\begin{align*}\beta_2=O(\min((k\log k)^{4-p}\log^{2p+2}n, (k\log k)^{4-2p}\log^{4+p} n)).\end{align*}
For the running time: $SA,T_1A$ can be computed in $\mathrm{nnz}(A)$ time. Thus, total running time is $\mathrm{nnz}(A)+(n+d)\mathrm{poly}(k).$
\end{proof}

\section{Implementation Setups}
We implement all the algorithms in MATLAB. We ran experiments on a machine with 16G main memory and Intel Core i7-3720QM@2.60GHz CPU. The operating system is Ubuntu 14.04.5 LTS. All the experiments were in single threaded mode.

\section{Data Simulation for Comparison with $\ell_{1}$ and $ \ell_{2}$ Regression}
We generate a matrix $A\in\mathbb{R}^{n\times d},x^*\in\mathbb{R}^d$ as following: set each entry of the first $d+5$ rows of $A$ as i.i.d. standard random Gaussian variable, each entry of $x^*$ as i.i.d. standard random Gaussian variable. For $n \geq i \geq d + 6$, we uniformly choose $p \in [d + 5]$, and set $A^{i} = A^{p}, b_{i} = b_{p}$.  We perform experiments under 3 different noise assumptions and 2 dimension combinations of $N, d$ and in total $3 \times 2 = 6$ experiments. The 3 different noise assumptions are, respectively i) $N(0,50)$ Gaussian noise with on all the entries of $Ax^*$; ii) sparse noise, where we randomly pick $3\%$ number of entries of $Ax^*$, and add uniform random noise from $[-\|Ax^*\|_2, \|Ax^*\|_2]$ on each entry to get $b$; %in addition, for $n \geq i \geq d + 6$, we uniformly choose $p \in [d + 5]$, and set $A^{i} = A^{p}, b_{i} = b_{p}$.
 iii) mixed noise, which is $N(0,5)$ Gaussian noise plus sparse noise. The $2$ different dimension combinations are i) balance, where $n = 100 \approx d = 75$; ii) overconstraint, where $n = 200 \gg d = 10$.

\section{Experiments on Approximation Ratio}
Here is a documentation of our preliminary experiments on calculating the actual approximation ratio for the experiment settings mentioned in \textbf{Section 5.1, Comparison with $\ell_{1}$ and $\ell_{2}$ regression}. The approximation ratio of interest is calculated as follows: $\frac{\|Ax'-b\|_G}{\|Ax^*-b\|_G}$, where $x'$ is the output of our novel embedding based algorithm and $x^*$ is the optimal solution. Since $\|\cdot\|_{G}$ is convex, we can formulate this problem as a convex optimization problem and use a vanilla gradient descent algorithm to calculate the optimal solution. We heuristically stop our gradient descent algorithm when the one step brings less than $10^{-7}$ improvement on the loss function and set the learning rate to be 0.001. Admittedly, we have not yet thoroughly and rigidly examined the convergence of the vanilla gradient descent algorithm (a direction of future work), and hence such calculation of approximation ratio is only a preliminary attempt.

Under the mixed noise setting, we varied different scale s of the uniform noise to be $0, 1, 2, 3$ and delta to be $0.1, 0.25, 0.5, 0.75$. With $n=200, d=10$, for each of these 4 * 4 = 16 settings, we run the algorithm repeatedly for 50 times, and the worst approximation ratio is 1.06 among these 800 runs. Experimentally, it is far below the theoretical guarantee $d\log^{2}(n) \approx 584 \gg 1.06$, and the approximation ratio is robust among different noise settings. For $n=100, d=75$, due to time limit, we only run each of the 16 settings for 5 times, and the worst approximation ratio is $1.31$.

\section{Implementation Detail for Low Rank Approximation}
\begin{itemize}
    \item For our algorithm, set $t_1=4k,t_2=8t_1$, set $S\in\mathbb{R}^{t_1\times n},T_1\in\mathbb{R}^{t_2\times n}$ to be two random cauchy matrices, and set $R\in\mathbb{R}^{d\times t_1},T_2\in\mathbb{R}^{d\times t_2}$ to be two embedding matrices with exponential random variables (see Theorem 16.) We solve the minimization problem $\min_{X,Y}\|T_1ARXYSAT_2-T_1AT_2\|_F^2$, and set $B=ARXYSA.$
    \item  For algorithm in~\cite{swz17}, we set $t_1=4k,t_2=8t_1$. We set $S\in\mathbb{R}^{t_1\times n},T_1\in\mathbb{R}^{t_2\times n},R\in\mathbb{R}^{d\times t_1},T_2\in\mathbb{R}^{d\times t_2}$ to be four random cauchy matrices. We solve the minimization problem $\min_{X,Y}\|T_1ARXYSAT_2-T_1AT_2\|_F^2$, and set $B=ARXYSA.$
    \item For PCA, we project $A$ onto the space spanned by top $k$ singular vectors to get $B$.
\end{itemize}

%% file: related.tex
\section{Related Works}

Existing literature studied the robust regression with respect to Huber loss function~\cite{MM00,Owen07}. Such regression can be applied to solve many problems like the people counting problem~\cite{CM15}.
To speed up the regression process,
some dimensional reduction techniques can be used to reduce the number of observations~\cite{GIMQS15}, also faster algorithms have been proposed to address the robust regression with reasonable assumption~\cite{bjk15}.
Besides, different models of regression were explored, such as Gaussian process regression~\cite{Rasmussen06}, active regression with adaptive Huber loss~\cite{CM16}.

Recent years, there are lots of randomized sketching and embedding techniques developed for solving numerical linear algebra problems. There is a long line of works, e.g.~\cite{a03,cw13,nn13} for $\ell_2$ subspace embedding, and works, e.g.~\cite{sw11,mm13,wz13,ww18} for $\ell_p$ subspace embedding. For more related works, we refer readers to the book~\cite{w14}. Based on sketching/embedding techniques, there is a line of works studied $\ell_2$ and $\ell_p$ regressions, e.g.~\cite{cw13,dmms11,mm13,nn13,wz13}. \cite{cw15} studied linear regression with M-estimator error measure. We refer to the survey \cite{m11} for more details.

Frobenius norm low rank matrix approximation problem is also known as PCA problem. This problem is well studied. The fastest algorithm is shown by~\cite{cw13}. For the entrywise $\ell_p$ norm low rank approximation problem, there is no known algorithm with theoretical guarantee until the work~\cite{swz17}. \cite{swz17} works only for $1\leq p\leq 2.$ Recently, ~\cite{cgklpw17} gives algorithms for all $p\geq 1.$ But either the running time is not in polynomial or the rank of the output is not exact $k$.